\documentclass{article} %
\usepackage{preprint}
\usepackage{times}
\iclrfinalcopy %

\usepackage{microtype}
\usepackage{amsmath,amssymb,amsfonts,amsthm}
\usepackage{booktabs}
\usepackage{tabularx}
\usepackage{hyperref}
\usepackage{xcolor}
\usepackage{xspace}
\usepackage{enumitem}
\usepackage{transparent}
\usepackage{listings}
\usepackage{dsfont}

\usepackage{algorithm}       %

\usepackage[noend]{algpseudocode}
\newcommand{\aim}{\texttt{AIM}\xspace}
\newcommand{\aimgdp}{\texttt{AIM-GDP}\xspace}
\newcommand{\pgm}{\textsf{PGM}}
\newcommand{\mixcal}{\texttt{MixCal}\xspace}
\newcommand{\dpsgd}{\texttt{DP-SGD}\xspace}

\newcommand{\set}[1]{\left\{#1\right\}}
\newcommand{\norm}[1]{\left\lVert#1\right\rVert}

\DeclareMathOperator*{\argmin}{arg\,min}

\usepackage{thmtools}
\usepackage{nicefrac}
\usepackage{inconsolata}
\usepackage[bb=dsserif]{mathalpha}
\usepackage[group-separator={,},group-minimum-digits={3}]{siunitx}
\usepackage{csquotes}
\usepackage{multirow}
\renewcommand{\cite}{\citep}

\usepackage{pgfplots}
\usetikzlibrary{fadings}
\usetikzlibrary{patterns}
\usetikzlibrary{shadows.blur}
\usetikzlibrary{shapes}
\usetikzlibrary{positioning}
\usetikzlibrary{arrows.meta}
\usetikzlibrary{fit}
\usetikzlibrary{calc}
\usepackage{subcaption}
\usepackage{wrapfig}

\usepackage{appendix}
\AtBeginEnvironment{appendices}{\crefalias{section}{appendix}}
\usepackage[capitalize]{cleveref}

\hypersetup{
    colorlinks=true,
    linkcolor=blue!70!black,
    citecolor=blue!70!black,
    urlcolor=blue!70!black
}

\DeclareMathOperator{\E}{\mathbb{E}}

\newcommand{\sD}{\mathbb{D}}

\newcommand{\sR}{\mathbb{R}}

\newcommand{\sY}{\mathbb{Y}}

\newcommand{\dd}{\mathrm{d}}

\newcommand{\bern}{\mathsf{Bern}}

\newcommand{\define}{~\smash{\triangleq}~}

\newcommand{\indic}{\mathds{1}}
\setitemize{itemsep=-.2em}

\definecolor{seabornblue}{HTML}{1f77b4}
\definecolor{seabornorange}{HTML}{ff7f0e}
\definecolor{seaborngreen}{HTML}{2ca02c}
\definecolor{seabornred}{HTML}{d62728}

\newtheorem{theorem}{Theorem}[section]
\newtheorem{proposition}{Proposition}[section]
\newtheorem{lemma}{Lemma}[section]
\newtheorem{corollary}{Corollary}[section]
\theoremstyle{definition}

\title{\centering Tight End-to-End Accounting for \\ Differentially Private Synthetic Data}

\usepackage{authblk}

\usepackage{authblk}

\author[1,2]{\bf Georgi Ganev}
\author[3]{\bf Vadym Doroshenko}
\author[4]{\bf Arun Ganesh}
\author[5]{\bf Juan Felipe Gomez}
\author[4]{\authorcr \bf Ryan McKenna}
\author[6]{\bf Bogdan Kulynych}

\affil[1]{\footnotesize UCL}
\affil[2]{\footnotesize SAS}
\affil[3]{\footnotesize Google}
\affil[4]{\footnotesize Google Research}
\affil[5]{\footnotesize Hasso Plattner Institute}
\affil[6]{\footnotesize Lausanne University Hospital (CHUV)}

\begin{document}

\maketitle

\begin{abstract}
Differentially private (DP) pipelines often combine heterogeneous mechanisms, yet their guarantees are best analyzed using different notions of privacy (e.g., bounded range, zCDP, ADP). When composing them, we have to analyze the guarantees through a reduced common representation that loses accounting precision.
This issue is particularly relevant in DP synthetic data generation, where data preprocessing, generation algorithms, and \dpsgd could appear within the same pipeline.
We develop novel tools for tight end-to-end accounting of such compositions using privacy loss distributions (PLD), $f$-DP, and Gaussian DP (GDP). In particular, we derive PLD and GDP analyses of bounded-range mechanisms such as the Exponential mechanism.
Building on these theoretical results, we introduce \emph{\aimgdp}, a GDP-native variant of \aim (a state-of-the-art DP synthetic data mechanism) that enables easy composability with complex pipelines via PLD, attack-aware risk calibration, and tight end-to-end auditing while consistently achieving better utility; and \emph{\mixcal}, an algorithm for calibrating adaptive mixed compositions using PLD accounting.
Overall, our tools enable tighter end-to-end accounting and thus less noise at the same privacy level in DP synthetic data pipelines.
\end{abstract}

\section{Introduction}
\label{sec:intro}

A core strength of differential privacy (DP)~\cite{dwork2006calibrating,dwork2014algorithmic} is its composability: we can construct highly complex randomized algorithms by combining multiple simpler primitives. At the same time, in modern data-analytic workflows, the components are highly heterogeneous, e.g., including different data preprocessing algorithms, and using combinations of learning algorithms.

Despite the composability of DP, there is a bottleneck when analyzing the privacy guarantees of such heterogeneous workflows. DP primitives have been developed using disparate analytical approaches, e.g., zero-concentrated DP (zCDP)~\cite{bun2016concentrated,dwork2016concentrated} for compositions of private selection and simple noise addition algorithms~\cite{dong2020optimal}, bounded range (BR)~\cite{durfee2019practical} for private selection, R\'enyi DP (RDP)~\cite{mironov2017renyi} for data-dependent compositions~\cite{papernot2018scalable} and early analyses of \dpsgd~\cite{abadi2016deep}, and privacy loss distributions (PLDs)~\cite{meiser2018tight} for more recent analyses of \dpsgd~\cite{koskela2020computing,gopi2021numerical,doroshenko2022connect}. Because of this, when composing multiple primitives analyzed using different approaches, we have to reduce the potentially rich representations to ``lowest common denominators,'' incurring loss in accounting precision.

As an example, consider \aim~\cite{mckenna2022aim}, a state-of-the-art algorithm for generating synthetic tabular data with DP. 
Internally, this algorithm is a fully adaptive composition~\cite{rogers2016privacy} of Exponential and Gaussian mechanisms~\cite{dwork2014algorithmic}.
The standard way to perform accounting for such compositions is using zCDP. This notion, although it is convenient and tightly characterizes the Exponential mechanism, does \emph{not} tightly characterize the trade-off curve~\cite{dong2019gaussian} of the Gaussian mechanism~\cite{zhu2022optimal,riess2026optimal}.
Because of this, the resulting privacy accounting is lossy both (1) within \aim itself, and (2) when \aim is further composed with \dpsgd for multimodal synthetic data generation~\cite{hu2026actg}.
In contrast, Gaussian DP (GDP)~\cite{dong2019gaussian} optimally characterizes the Gaussian mechanism, and would compose well with \dpsgd or other pipelines using PLD accounting, but there are no general analyses of the Exponential mechanism in terms of GDP that do not require a special structure~\cite{gopi2022private}.

In this work, we improve the accounting precision by developing methods that are more suited for end-to-end accounting of realistic synthetic data workloads.
Our core contributions are:

\begin{table*}[t!]
\vspace{-10pt}
\centering
\scriptsize
\setlength{\tabcolsep}{5pt}

\begin{minipage}[c]{0.4\textwidth}
\centering
\caption{Novel theoretical results}
\label{tab:theory}
\vspace{-7pt}

\begin{tabular}{clll}
\toprule
& \textbf{Mechanism} & \textbf{Analysis} & \textbf{Ref.} \\
\midrule
\textbf{i}
& Bounded range
& $f$-DP/GDP
& Sec.~\ref{sec:technical-results}
\\

\textbf{ii}
& Bounded range
& PLD
& Sec.~\ref{sec:technical-results}
\\

\textbf{iii}
& Partition selection
& PLD
& App.~\ref{app:extra-methods}
\\
\textbf{iv}
& Poisson subsampling
& PLD
& App.~\ref{app:extra-methods}
\\
\textbf{v}
& Group privacy
& Exponential
& App.~\ref{app:extra-methods}
\\
\bottomrule
\end{tabular}
\end{minipage}
\hfill
\begin{minipage}[c]{0.56\textwidth}
\centering
\caption{Specific applications and algorithms/pipelines}
\label{tab:applications}
\vspace{-7pt}

\begin{tabular}{llc}
\toprule
\textbf{Application}
& \textbf{Algorithm/pipeline}
& \textbf{Can use} \\
\midrule

Discrete synthesis
& \aimgdp
& i
\\

PLD calibration
& \mixcal
& ii--iii
\\

End-to-end synthesis
& \aimgdp + DP preproc. + \mixcal
& i--v
\\

Multimodal synthesis
& \aimgdp + DP-SGD + \mixcal
& i--v
\\

\bottomrule
\end{tabular}
\end{minipage}

\end{table*}

\begin{figure*}[t]
    \vspace{-10pt}
    \centering
    \begin{subfigure}[b]{0.32\linewidth}
        \centering
        \raisebox{3pt}{
        \includegraphics[width=0.99\linewidth]{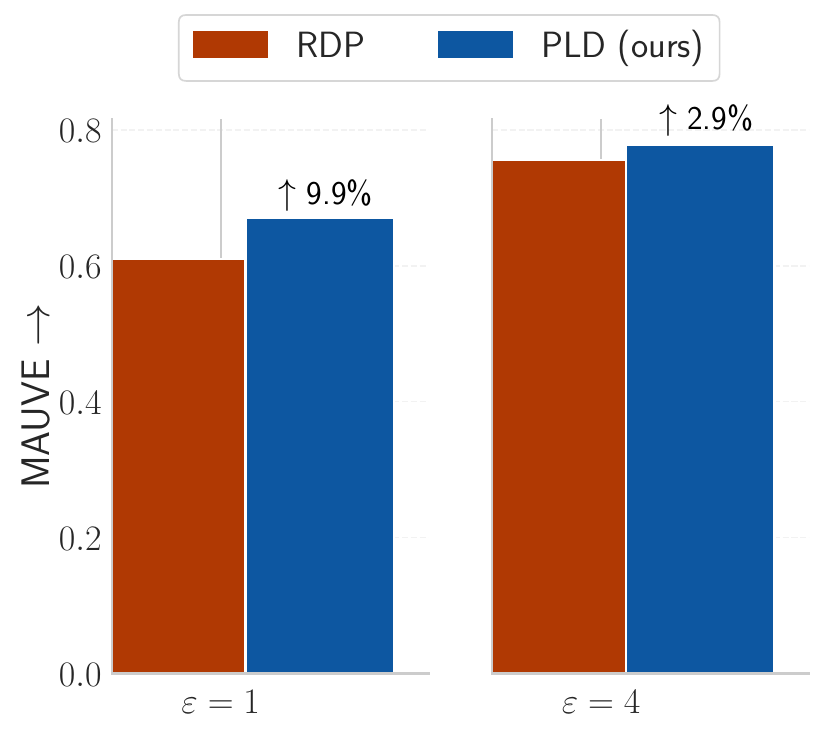}
        }
        \caption{Improved utility}
        \label{fig:utility}
    \end{subfigure}
    \hfill
    \begin{subfigure}[b]{0.32\linewidth}
        \centering
        \raisebox{5pt}{
        \includegraphics[width=0.99\linewidth]{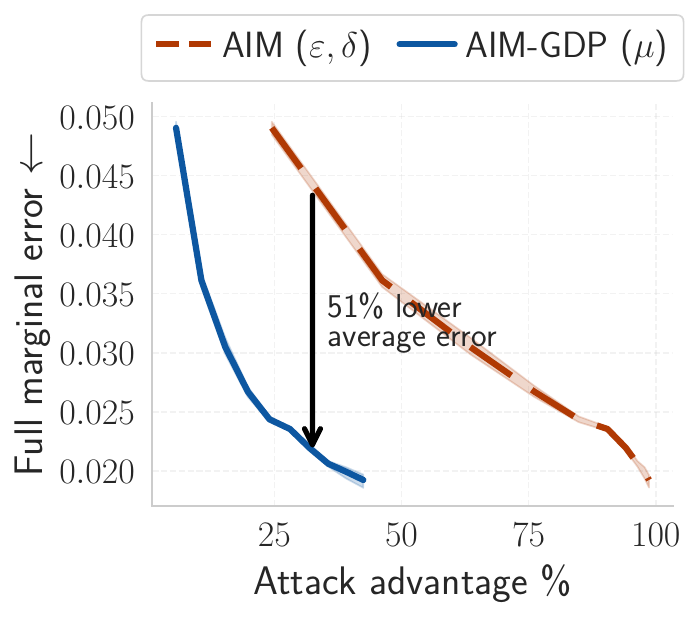}
        }
        \caption{Attack-aware calibration}
        \label{fig:attack}
    \end{subfigure}
    \hfill
    \begin{subfigure}[b]{0.32\linewidth}
        \includegraphics[width=0.99\linewidth]{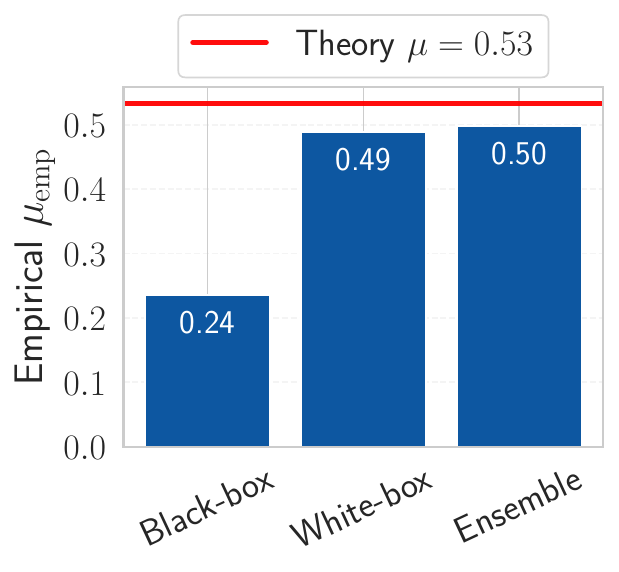}
        \caption{First end-to-end audit}
        \label{fig:audit}
    \end{subfigure}
    \caption{\textbf{Our methods provide improved utility and enable new capabilities in DP synthetic data workflows.} \emph{Top:} Summary of theoretical results and algorithmic contributions. \emph{Bottom:} (a) Our methods improve utility in synthetic text generation using LLMs. (b) Our methods enable to efficiently calibrate DP workflows to attack risk, with \emph{the first application of attack-aware calibration to synthetic data.} (c) Our methods support tight end-to-end privacy audits, with \emph{the first end-to-end audit of the synthetic data generator,} \aimgdp, our version of \aim.}
    \label{fig:fig1}
    \vspace{-10pt}
\end{figure*}

\begin{enumerate}[leftmargin=*]
    \item \textbf{PLD, $f$-DP, and GDP analyses of BR mechanisms for private selection.} Using these results, we show that GDP enables more precise analyses than zCDP in fully adaptive compositions where most of the privacy budget is spent on Gaussian mechanisms (as in \aim).
    We provide additional results that enable the support of more workloads, e.g., involving private partition selection, data subsampling, and multiple user contributions in \cref{app:extra-methods}. See a summary in \cref{tab:theory}.
    
    \item \textbf{\aimgdp.} Building on the GDP analysis of the Exponential mechanism, we propose \emph{\aimgdp}, a variant of \aim that natively satisfies GDP. It provides comparable/better utility (\cref{fig:utility}) with additional benefits. First, \aimgdp enables the \textbf{first efficient attack-aware calibration for DP synthetic data}~\cite{kulynych2024attack,kulynych2025unifying}, enabling practitioners to set privacy budgets based on target risk levels (\cref{fig:attack}). Second, it supports tight audits~\cite{nasr2023tight,annamalai2024what}, making it the \textbf{first end-to-end audited} implementation of \aim (\cref{fig:audit}). Third, it is \textbf{easily composable} in complex pipelines, as GDP is compatible with tight PLD composition analyses.

    \item \textbf{\mixcal.} We introduce \emph{\mixcal}, a new algorithm for calibrating parameters in mixed-mechanism adaptive compositions that combines the advantages of zCDP or GDP for simple calibration of parameters across multiple components, and of PLD accounting for precision.
    Building on this, we extend \aimgdp to more complex pipelines by combining it with DP preprocessing for end-to-end synthetic data generation and with \dpsgd for multimodal synthesis. See a summary in \cref{tab:applications}.
    
    \item We demonstrate the combinations of these algorithms in realistic synthetic data workflows, showing both performance improvements at the same privacy level and new qualitative capabilities. 
    Additionally, we release open-source implementations of our main methods (e.g., \aimgdp, \mixcal) in \url{https://github.com/google/dpsynth} under Apache 2.0 License.

\end{enumerate}

\section{Problem Statement}
\label{sec:background}

\subsection{Motivation and Objective}

Our main motivating setting is private tabular or multimodal synthetic data generation, even though, as we show next, our technical results are relevant to any mixed-mechanism compositions.

\begin{wrapfigure}[11]{r}{0.39\textwidth}
  \centering
  \vspace{-15pt}
  \includegraphics[width=0.99\linewidth]{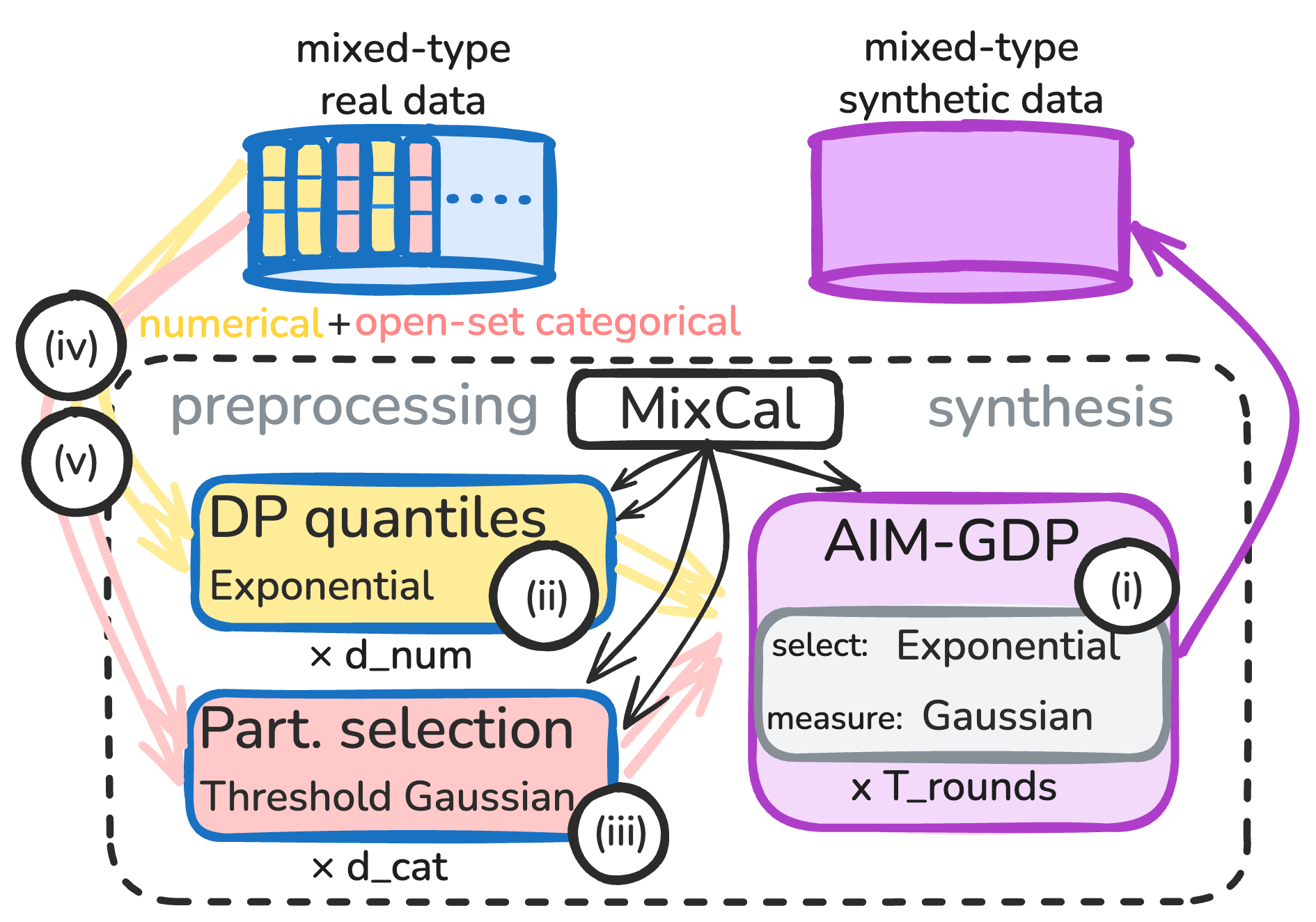}
  \vspace{-15pt}
  \caption{End-to-end synthesis.}
  \label{fig:c_pipeline}
\end{wrapfigure}

\emph{Discrete data.}
For synthetic tabular data, the best-performing DP mechanisms are marginal-based~\cite{tao2022benchmarking,chen2025benchmarking}; which typically operate on discrete domains, and consist of a combination of selection steps (e.g., using the Exponential mechanism), and measurement steps (e.g., using the Gaussian mechanism). Some mechanisms (like \aim) further rely on \emph{fully-adaptive} composition, where privacy budget assigned to each Exponential and Gaussian mechanism is a data-dependent random variable.

\emph{End-to-end data.}
Beyond discrete data, real-world datasets often contain a mix of categorical attributes with potentially unknown domains (open-set attributes) and numerical attributes. To leverage marginal-based mechanisms, we must map our data to a discrete domain.
For numerical attributes, data-dependent discretization, e.g., via DP quantiles~\cite{kaplan2022differentially}, is a preferred approach~\cite{ganev2025importance}.
For open-set categorical attributes, private domain discovery~\cite{gopi20setunion} poses an additional challenge, and, to the best of our knowledge, existing synthetic data implementations lack such support. \cref{fig:c_pipeline} shows an example pipeline, which we discuss in detail in \cref{sec:applications}. 

\emph{Multimodal data.} Beyond tabular data, generating synthetic \emph{unstructured} data usually involves using a \dpsgd-style mechanism (e.g., finetuning an LLM) \cite{ponomareva2026dp}.  ACTG, a state-of-the-art mechanism for this task consists of a synthetic tabular generation step plus \emph{conditional} finetuning step \cite{hu2026actg}. This work relies on RDP accounting due to the reliance on \aim that only provides a zCDP guarantee.

We provide a review of related mechanisms for data preprocessing and synthetization used in DP synthetic data pipelines in \cref{tab:mechanism_taxonomy} in the Appendix.

\textbf{Our goal.} The objective of our work is to provide methods for tight end-to-end privacy analyses for such heterogeneous pipelines. 
We observe two key challenges: (1) we lack precise analyses of mixed-mechanism compositions involving the Exponential mechanism for both fully adaptive and adaptive settings, and (2) it is unclear how to calibrate the parameters of the component mechanisms in such composed pipelines. In the remainder, we introduce several fundamental results and new algorithms to solve these challenges.

\subsection{Technical Background}
We provide a concise overview of the technical notions we use. For a more in-depth treatment, we refer the reader to \citet{zhu2022optimal}, \citet{gomez2026gaussian} or any of the cited works in this section.
Throughout, $P$ and $Q$ are probability distributions on a measurable space $\sY$. In an abuse of notation, we also write $P(y)$ and $Q(y)$ to denote 
their densities with respect to a common $\sigma$-finite measure, e.g., $P + Q$.
We write $(x)_+ \define \max\{x, 0\}$ and use the conventions $\log(\nicefrac{a}{0}) = +\infty$ and $\log(\nicefrac{0}{a}) = -\infty$ for $a > 0$, and $e^{-\infty} = 0$. 

We begin by defining two divergences: the hockeystick divergence of order $e^\varepsilon > 0$ and the R\'enyi divergence of order $\tau > 1$, respectively:
$$H_{e^\varepsilon}(P \parallel Q) \define \sup_{E \subseteq \sY}(P(E) - e^\varepsilon Q(E)), \quad D_\tau(P \parallel Q) \define \tfrac{1}{\tau-1}\log \E_{y \sim Q}\left(\tfrac{P(y)}{Q(y)}\right)^\tau.$$
With these divergences, we can define two notions of differential privacy (DP). A mechanism $M(\cdot)$ satisfies $(\varepsilon, \delta)$-DP if $H_{e^\varepsilon}(M(S) \parallel M(S')) \leq \delta$ for any \emph{neighbouring dataset} pair $S \simeq S'$. 
A mechanism satisfies $(\tau, \varepsilon_\text{RDP})$-R\'enyi DP (RDP) if $D_\tau(M(S) \parallel M(S'))\leq \varepsilon_\text{RDP}$. 
Zero-concentrated DP (zCDP)~\cite{dwork2016concentrated,bun2016concentrated} is a special case of RDP: a mechanism is $\rho$-zCDP if it is $(\tau, \tau \cdot \rho)$-RDP for any $\tau > 1$. 

The function $\delta: \sR \to [0,1]$ such that a mechanism satisfies $(\varepsilon, \delta(\varepsilon))$-DP for any $\varepsilon \in \sR$ is called a \emph{privacy profile}~\cite{balle2018privacy}, and is closely related to the operational formulation of DP known as $f$-DP~\cite{dong2019gaussian}. 
Let $T(P \parallel Q)$ denote the trade-off function which maps the Type I error (denoted by $\alpha$) of the most powerful test distinguishing $H_0 \sim P$ and $H_1 \sim Q$ to its Type II error (denoted by $\beta$). We say a mechanism satisfies $f$-DP if $f = T(P \parallel Q)$ for some pair $(P,Q)$ and $f \leq T(M(S) \parallel M(S'))$ for all $S \simeq S'$. 
An important special case of $f$-DP is Gaussian DP (GDP). We say a mechanism is $\mu$-GDP if it satisfies $f_\mu$-DP, with $f_\mu(\alpha) = \Phi(\Phi^{-1}(1 - \alpha) - \mu)$, where $\Phi(t)$ is the CDF of the standard normal distribution~\cite{dong2019gaussian}. 

The celebrated Gaussian mechanism~\citep{dwork2014algorithmic}, which adds Gaussian noise $\mathcal{N}(0, \sigma^2 \cdot I_d)$ to a vector query $q$ with $\| q(S) - q(S') \|_2 \leq 1$ for any pair $S \simeq S'$, is $\mu $-GDP with $\mu = \nicefrac{1}{\sigma}$. The Gaussian mechanism also satisfies $\rho$-zCDP with $\rho = \nicefrac{\mu^2}{2}$. 

One of the most fundamental properties about differential privacy is \emph{composition}, where the outputs of one mechanism are used as the input of another. For example, the composition of $M_1$ and $M_2$ is $M_\text{composed} = M_1(D, M_2(D))$. 
We refer to $M_1$ and $M_2$ as \emph{primitive} mechanisms.
We distinguish between two types of composition: the first, which we call \emph{adaptive composition} or simply \emph{composition}, is when the number of primitive mechanisms and the choice of all $(\varepsilon,\delta)$ is fixed ahead of time. The second type, which we call \emph{fully adaptive composition}, occurs when the number of primitive mechanisms is not known, and only a fixed budget $(\varepsilon, \delta)$ is known ahead of time~\cite{rogers2016privacy}. 

RDP admits a simple composition rule: if $M_1$ satisfies $(\tau, \varepsilon_\text{RDP})$-RDP and $M_2$ satisfies $(\tau, \varepsilon_\text{RDP}')$, then the composed mechanism satisfies $(\tau, \varepsilon_\text{RDP}+\varepsilon_\text{RDP}')$-RDP. For $\rho$-zCDP, the composition rule is particularly simple: sum the $\rho$ of the components. This RDP composition result holds for \emph{both} adaptive and fully-adaptive composition~\citep{feldman2021individual}. GDP composition also admits a simple form $\smash{\sqrt{\mu_1^2 + \mu_2^2}}$, and this  holds in fully-adaptive setting~\cite{smith2022fully,koskela2023individual}. Tight composition results for $(\varepsilon,\delta)$-DP and its generalizations (i.e. privacy profiles, $f$-DP) are not as simple, and there are different results for adaptive and fully-adaptive composition. 

For adaptive composition, tight composition results can be derived via dominating pairs~\cite{zhu2022optimal}. A pair of distributions $P, Q$ is a \emph{dominating pair} for a mechanism $M(\cdot)$ if for any $\varepsilon \in \sR$, we have $H_{e^\varepsilon}(P\parallel Q) \geq \sup_{S \simeq S'}H_{e^\varepsilon}(M(S) \parallel M(S'))$. We say $(P,Q)$ is a tightly dominating pair if the inequality is equality. 
Given a dominating pair $(P,Q)$, composition becomes simple: if $M_1$ has dominating pair $(P_1, Q_1)$ and $M_2$ has dominating pair $(P_2, Q_2)$, then the composed mechanism has dominating pair $(P_1 \otimes P_2, Q_1 \otimes Q_2)$. 

\emph{Privacy Loss Distributions} (PLDs) and \emph{Privacy Loss Random Variables} (PLRVs)~\citep{sommer2019privacy} make this analysis even simpler. Given a dominating pair $(P,Q)$, define the random variable $\omega_{P/Q} \define \log \nicefrac{P(y)}{Q(y)}$ with $y \sim P$ as the privacy loss random variable, and, its distribution as the PLD. If $M_1$ has PLRV $\omega_{P_1/Q_1}$ and $M_2$ has PLRV $\omega_{P_2/Q_2}$, then the composed mechanism has PLRV $\omega_{P_1/Q_1} + \omega_{P_2/Q_2}$. Because this addition is over independent random variables, the composed mechanism's PLD is a convolution of the individual PLDs. Given a PLD, we can compute the privacy profile as $\delta(\varepsilon) = \E[\smash{(1 - e^{\varepsilon - \omega_{P/Q}})_+}]$. 
Analogously to dominating pairs, we say that a PLD $\omega$ \emph{dominates} a mechanism $M(\cdot)$ if its privacy profile is an upper bound of the actual privacy profile of the mechanism for all $\varepsilon \in \sR$.
PLDs provide a recipe for tight composition of mechanisms: first, compute the PLD of each primitive, then convolve their PLDs and compute the composed profile or the trade-off function from the convolved PLD.

For fully-adaptive composition, the PLD approach generally does not provide valid bounds~\cite{regehr2026privacy, tran2026f}, and, as far as we are aware, the tightest known bounds for general heterogeneous fully-adaptive composition of mechanisms are given by $\rho$-zCDP or $\mu$-GDP.

We conclude this section with one last privacy representation: \emph{bounded range} (BR)~\citep{durfee2019practical}.
We say a mechanism $M:2^\sD \to \sY$ satisfies $\nu$-bounded range (BR) if for any $S \simeq S'$ there exists $0 \leq t \leq \nu$ such that for all measurable $E \subseteq \sY$, the range of privacy loss is bounded:
\begin{equation}\label{eq_def_br}
        \log\left( \frac{\Pr[M(S) \in E]}{\Pr[M(S') \in E]} \right) \in [-t, -t + \nu].
\end{equation}

Bounded range is especially suitable to characterizing the standard \emph{Exponential mechanism}. An Exponential mechanism $M(\cdot)$ with a discrete output space $\sY$ is a mechanism with probability distribution  
$
    \smash{\Pr[M(S) = y] \propto \exp\left(-\nicefrac{\nu}{2} \cdot q(S, y) \right)},
$
where $q: 2^\sD \to \sR$ is a query such that for any $S \simeq S'$ and $y \in \sY$, we have $|q(S, y) - q(S', y)| \leq 1$. Exponential mechanism satisfies $\nu$-BR. Notably, it also satisfies zCDP with $\rho = \nicefrac{1}{8} \cdot \nu^2$~\cite{durfee2019practical}, with a tighter expression recently given by \citet{harrison2025exact}.

\section{New Building Blocks for Heterogeneous Accounting}
\label{sec:technical-results}

In this section, we provide several new theoretical results that we later use as building blocks.

\paragraph{Novel analyses of Bounded Range mechanisms.}\label{sec:br}
We introduce a tight characterization of BR mechanisms by deriving their privacy profile, trade-off function, and PLD. From the trade-off function, we also derive a tight GDP bound. The latter is useful for analyzing \emph{fully adaptive compositions} of mechanisms involving BR, such as the \aim algorithm for synthetic data generation.

\begin{restatable}{theorem}{brprivprof}\label{stmt:br-to-pld}
    Let $M(\cdot)$ satisfy $\nu$-BR. Then $M(\cdot)$ satisfies $f$-DP and
    $(\varepsilon, \delta(\varepsilon))$-DP for all $\varepsilon \in \sR$, with
    \begin{equation}
        f(\alpha) = \frac{1 - \alpha}{\alpha e^\nu + 1 - \alpha}, \quad
        \delta(\varepsilon) = \begin{cases}
        1 - e^{\varepsilon}, & \varepsilon \le -\nu, \\
        \frac{\left(e^{\nu/2} - e^{\varepsilon/2}\right)^2}{e^{\nu} - 1}, & -\nu < \varepsilon < \nu, \\
        0, & \varepsilon \ge \nu.
        \end{cases}
    \end{equation}
    Moreover, $M(\cdot)$ is dominated by the PLD with density:
    \begin{equation}
        g_\nu(x) = D_\nu\, e^{x/2}\, \indic_{[-\nu, \nu]}(x), \quad \text{where } D_\nu = \tfrac{1}{2}\left(e^{\nu/2} - e^{-\nu/2}\right)^{-1}.
    \end{equation}
    These bounds are tight over the class of $\nu$-BR mechanisms: for every $\alpha \in [0, 1]$
    (respectively, every $\varepsilon \in \sR$) there exist a $\nu$-BR mechanism and a pair
    $S \simeq S'$ that attain $f(\alpha)$ (respectively, $\delta(\varepsilon)$). 

\end{restatable}

\begin{wrapfigure}{r}{0.4\linewidth}
    \vspace{-1.5em}
    \centering
    \resizebox{\linewidth}{!}{\input{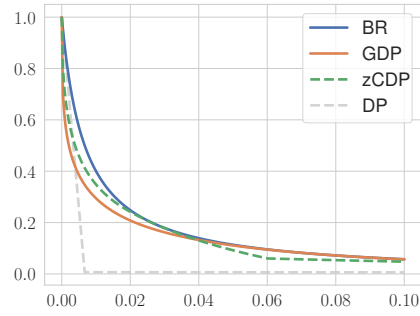}}
    \caption{Trade-off curves of a BR mechanism (higher is more precise) with $\nu = 4$, given by our $f$-DP  (\cref{stmt:br-to-pld}) and GDP analyses (\cref{stmt:br-to-gdp}), and zCDP using the recent optimal conversion~\cite{harrison2025exact}. To show the trade-off curves given by zCDP we use the conversion from RDP to $f$-DP~\cite{riess2026optimal}, for orders $> 1$.}
    \label{fig:trade-off-curves}
    \vspace{-2em}
\end{wrapfigure}

\begin{restatable}{theorem}{brtogdp}\label{stmt:br-to-gdp}
    Suppose that a mechanism $M(\cdot)$ satisfies $\nu$-BR. Then, it also satisfies $\mu$-GDP with:
    \begin{equation}\label{eq:br-to-gdp}
        \mu = -2 \Phi^{-1}\left(\frac{1}{e^{\nu/2} + 1}\right).
    \end{equation}
\end{restatable}

We provide the proofs of both statements and some discussion of them in \cref{app:proofs}.
There are various ways of comparing our new bounds to the previous tightest known zCDP bounds. We choose to compare the trade-off functions due to their interpretability.
We show the optimal BR trade-off function, our derived GDP trade-off function, and the tightest known trade-off function implied by zCDP in \cref{fig:trade-off-curves}. We observe that our optimal trade-off function provides a tighter bound than zCDP in all regimes, and our GDP bound provides a tighter bound than zCDP in some, but not all regimes. 

As we consider compositions of BR mechanisms, it is natural to ask if our analyses actually yield better results than zCDP. In practice, we care about compositions of mixed-mechanisms, often compositions of Gaussians and Exponential mechanisms. \cref{fig:gdp-demo} illustrates that our GDP analysis is tighter for mixed-mechanism compositions containing BR mechanisms where the majority of the budget is spent in Gaussian mechanisms, which is the setting of \aim, and we leverage this property next; and the PLD analysis is tighter for non-fully adaptive compositions, which is the setting of preprocessing workloads.

\begin{figure}
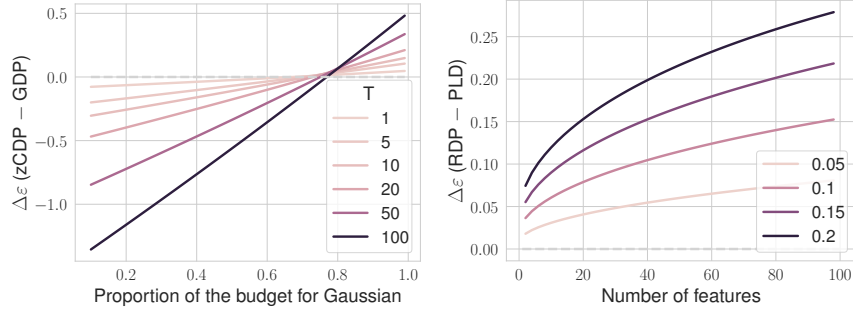

    \centering
    \resizebox{0.4\linewidth}{!}{\input{figures/gdp_benefit.pgf}}~
    \resizebox{0.4\linewidth}{!}{\input{figures/pld_benefit.pgf}}
    \caption{\textbf{Our methods enable tighter accounting in the regimes of interest.} On the y axis, we show the $\varepsilon$ decrease ($> 0$ means our methods are better) at $\delta = 10^{-5}$. \textit{Left:} Using GDP, as a function of the budget proportion of Gaussian mechanisms (x axis) and number of iterations (hue) in \aim-like workloads containing fully adaptive compositions of Exponential and Gaussians. \textit{Right:} Using PLD, as a function of the number of features (x axis) and the per-feature privacy parameter $\mu$ in data preprocessing-like workloads containing compositions of Exponential and Threshold Gaussian mechanisms. Details in \cref{app:extra-exp}.
    } 
    \vspace{-1em}
    \label{fig:gdp-demo}
\end{figure}

\paragraph{Subsampling, group privacy, tight partition selection.}
We provide a number of additional analyses for mechanisms of interest in various realistic data generation pipelines. All of these results are deliberately designed to be compatible with open-source software PLD libraries such as \texttt{dp\_accounting}. Due to space constraints, we defer the details to \cref{app:extra-methods}.

\emph{Poisson subsampling of arbitrary mechanisms.} We provide a method for converting a PLD of an arbitrary mechanism to the PLD for the mechanism wrapped with Poisson subsampling~\cite{abadi2016deep}. Previously, only additive mechanisms like Gaussian and Laplace mechanisms were supported. By enabling Poisson subsampling for arbitrary mechanisms, it becomes possible to scale synthetic data workflows to large datasets.

\emph{Group privacy for Exponential and other mechanisms on partitioned datasets.} We show that partitioning a dataset and then running, in particular, a $\nu$-BR mechanism on each subset of the partition satisfies $k\nu$-bounded range for groups of $k$ records. This enables a privacy analysis of mechanism such as DP quantiles~\cite{kaplan2022differentially} even when users can contribute multiple records. 

\emph{Characterization of Threshold Gaussian and Threshold Laplace.} The Threshold Gaussian and Laplace mechanisms~\cite{gopi20setunion} are useful mechanisms for DP partition selection. We provide a tighter characterization of their privacy guarantees in terms of PLD and privacy profiles.

\section{Applications}
\label{sec:applications}
Next, we use the results in \cref{sec:technical-results} to develop better algorithms for heterogeneous data analysis.
\paragraph{Discrete synthetic data with \aimgdp.}
\label{sec:gdpaim}
The developments in \cref{sec:br} can be used to improve the privacy analysis of any mechanism that uses the Exponential mechanism as a primitive. One such example is \aim~\cite{mckenna2022aim}, a state-of-the-art mechanism for synthetic data generation that operates within the select-measure-generate paradigm. \aim selects the worst-approximated marginal using Exponential mechanism, ``measures'' the selected marginal query using Gaussian mechanism, and repeats these two steps until the budget is depleted. Afterwards, a synthetic dataset is generated based on the measured marginals about real data. \aim is designed to spend the majority of the budget on the marginal measurements (Gaussian mechanism) because the signal is divided among the histogram bins and is hence more sensitive to noise. Whereas the original algorithm uses zCDP, our new GDP analysis in \cref{sec:br} provides improvements in exactly the regime where the majority of the budget in a mixed-mechanism composition is spent in Gaussian mechanisms. As a result, a GDP analysis of \aim is tighter than the original zCDP analysis. We present algorithms for both the original \aim, and our version, \aimgdp, in \cref{app:extra-details}. 

\paragraph{End-to-end synthetic data with \aimgdp, DP preprocessing, and \mixcal.}
\label{sec:mixed_type_synthesis}
One downside of \aim is that it can only be run on datasets with closed-domain or discrete attributes. In contrast, real-world tabular datasets typically contain a mix of numerical attributes and open-set categorical attributes with unbounded vocabularies (\cref{fig:c_pipeline}). Applying \aim to real datasets thus requires preprocessing these attributes with DP. As this is done on top of \aim without referencing the privacy budget, preprocessing workloads can be thought of as primitives that we compose \emph{non-}fully adaptively with \aim. For numerical attributes, 
we focus on using DP quantiles via recursive bisection~\cite{kaplan2022differentially} (see \cref{alg:quantiles} in the Appendix). Because this algorithm is a composition of standard Exponential mechanisms, it directly benefits from our PLD characterization of BR. For open-set categorical attributes, existing implementations either perform data-dependent domain construction without accounting for it~\cite{qian2023synthcity}, or do not support it at all~\cite{opendp2021smartnoise, mahiou2025dpmm}. We introduce the first principled integration of private partition selection, e.g., the Threshold Gaussian mechanism~\cite{gopi20setunion} (see \cref{alg:open_set} in the Appendix), into synthetic data pipelines, using our new PLD characterization to tightly account for domain discovery. Concretely, given a tabular dataset with numerical and open-set categorical attributes, first, we apply DP Quantiles and the Threshold Gaussian mechanism to map these attributes to discrete ones, and then, apply \aimgdp to generate synthetic data. The Gaussian PLD from \aimgdp, the PLDs of private set partition via Threshold Gaussian, and of discretization via DP quantiles (we introduce both in \cref{app:extra-methods}), are composable in a straightforward way using standard PLD accountants. Taken together, these results extend the usage of \aim to a larger class of tabular datasets. Moreover, in such workflows, different components might measure the same queries. In \cref{app:extra-details}, we also describe how to enable the re-use of such query outputs.

\emph{Calibrating the components.} In practice, however, there is a salient challenge when using PLD accounting in such workflows. We normally have a target guarantee, e.g., $(\varepsilon^\star, \delta^\star)$, or $(\alpha^\star, \beta^\star)$ for $f$-DP based attack-aware calibration (see \cref{sec:exp}). Therefore, we require a method to assign noise parameters (e.g., Gaussian parameter $\sigma$, Exponential Mechanism parameter $\nu$) to the primitive mechanisms such that the composed PLD satisfies the target guarantee.
A key challenge of this approach is that there is no clear way to decompose our target such that $\varepsilon_{\text{preprocess}} + \varepsilon_{\text{AIM}} = \varepsilon^\star$. This makes it difficult to ``assign'' more privacy budget to \aim and less to preprocessing, which is often desirable in practice. 

To solve this problem, we introduce a ``dummy'' privacy budget, denoted as $b$. As a simple example, suppose we are handling a composition of Gaussians, Exponential Mechanisms, and \aimgdp, all of which have a known PLD and $\rho$-zCDP guarantee. Starting from, e.g., the target $(\varepsilon^\star, \delta^\star)$, we use a conversion~\cite[e.g.,][]{bun2016concentrated} to zCDP to obtain the dummy privacy budget $b = \rho(\varepsilon^\star, \delta^\star)$. If we have a weight $w_1$ for preprocessing and weight $w_2$ for \aimgdp such that $w_1 + w_2 = 1$, we calibrate the preprocessing noise parameters to satisfy $w_1 \cdot b$-zCDP budget and calibrate \aimgdp to satisfy $w_2 \cdot b$-zCDP budget. We then use PLD accounting to analyze the composed mechanism. 
As our calibration itself is not precise, we would add more noise than necessary.
We thus perform binary search over $b$: if the resulting $\varepsilon$ is larger than the target $\varepsilon^\star$, we decrease $b$. Otherwise, we increase it. This allows us to meaningfully assign more or less privacy budget to preprocessing vs. \aimgdp, while making use of the tighter analysis of PLD over zCDP. 
This algorithm can be generalized to use any ``dummy'' budget that is additive under composition, e.g., GDP, and to variable numbers of compositions. We call this algorithm \mixcal and present it in \cref{alg:mixcal} in the Appendix.

\begin{wrapfigure}[9]{r}{0.33\textwidth}
  \centering
  \vspace{-20pt}
  \includegraphics[width=0.99\linewidth]{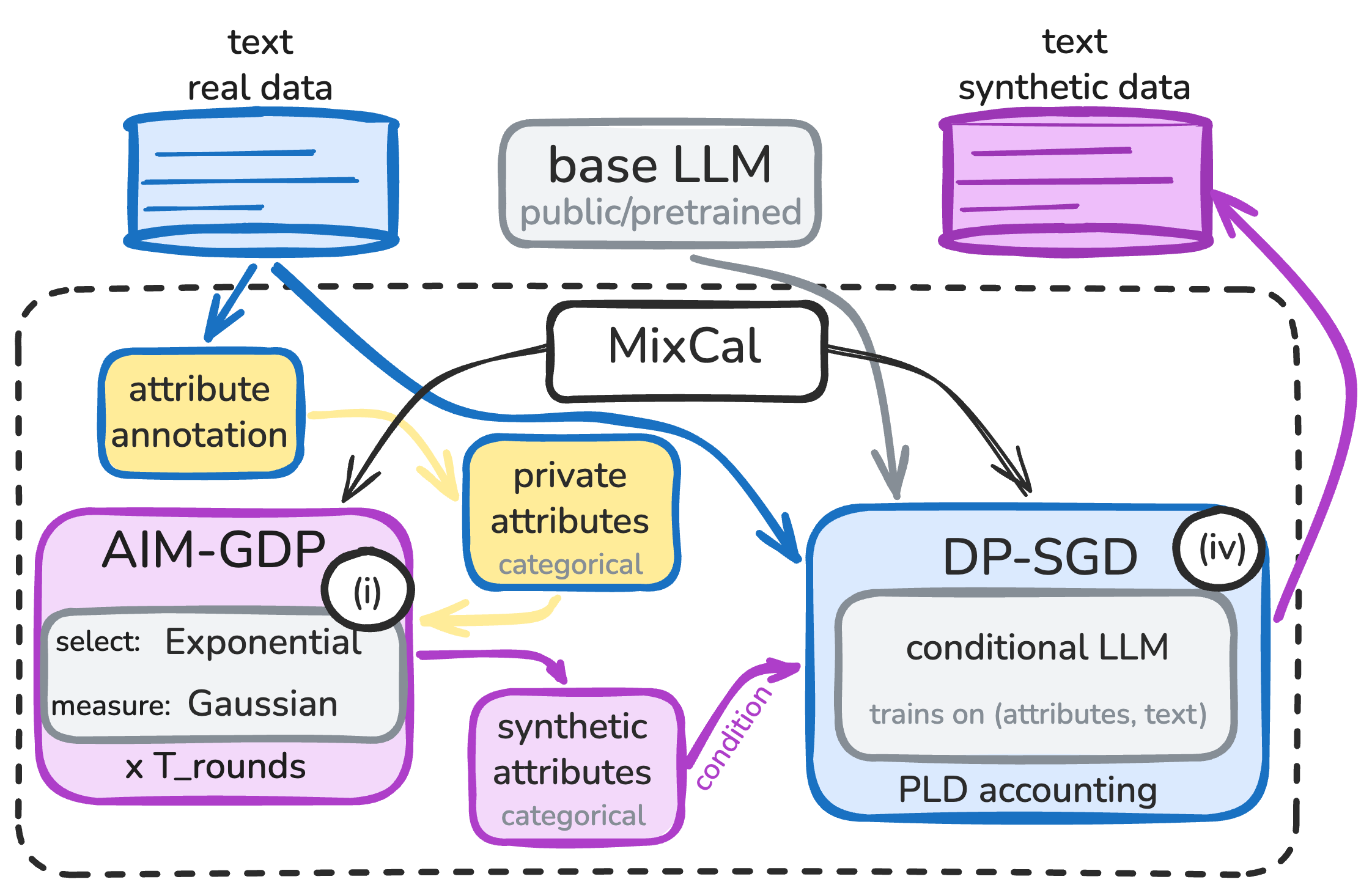}
  \vspace{-15pt}
  \caption{Multimodal synthesis.}
  \label{fig:d_pipeline}
\end{wrapfigure}
\paragraph{Multimodal synthetic data with \aimgdp, \dpsgd, and \mixcal.}
Recent work in synthetic text data generation with LLMs used \aim as a building block. For instance, the ACTG mechanism~\cite{hu2026actg} consists of several steps: (1) annotate input text data using a predefined tabular schema, (2) run \aim on the generated tabular data with a fraction of the budget, and (3) run \dpsgd to finetune a language model to predict the text given the tabular metadata (see \cref{fig:d_pipeline}).  The privacy analysis is hence a composition of \aim with \dpsgd, and, because \aim satisfies zCDP, \citet{hu2026actg} used RDP instead of the tighter PLD accounting. By replacing \aim with \aimgdp, we can analyze the algorithm end-to-end with PLD, and use \mixcal to allocate the budget. Beyond precision, PLD accounting enables more capabilities, such as using random allocation~\cite{feldman2025privacy} instead of Poisson subsampling in \dpsgd, which provides better utility in most regimes and is only analyzable using PLD, and subsampling large datasets using our new analyses in \cref{app:extra-methods}.

\section{Experimental Evaluation}
\label{sec:exp}

We now evaluate the algorithms introduced in \cref{sec:applications} across three synthetic data generation applications: discrete, end-to-end, and multimodal.

\subsection{Discrete Synthetic Data with \aimgdp}
\label{sec:exp-aimgdp}

We evaluate \aimgdp on four standard discrete datasets (with known categorical features; see ~\cref{tab:datasets} in the Appendix) across four privacy budgets, $\varepsilon \in \{0.5, 1, 2, 5\}$, $\delta = 10^{-5}$, with five runs per setting.
We assess utility using the mean TV distance between the marginal distributions of real data and model projections, computed over 1/2/3-way and full workloads.
We also demonstrate two additional capabilities enabled by \aimgdp: attack-aware risk calibration, and tight end-to-end privacy auditing. Note that these capabilities are also applicable to the entire pipelines analyzable via PLD using \mixcal. We focus on their evaluation for a standalone \aimgdp instance for simplicity.

\begin{table}[t]
\footnotesize
\centering
\caption{\textbf{\aimgdp consistently shows better utility.} We show
1/2/3-way and full marginal errors, number of measured marginals, and runtime, averaged across four datasets and four privacy budgets.
}
\label{tab:aimgdp_utility}
\resizebox{0.85\linewidth}{!}{
\begin{tabular}{ll cccc cc}
\toprule
&
\textbf{Algorithm}
& \textbf{1-way}
& \textbf{2-way}
& \textbf{3-way}
& \textbf{Full}
& \textbf{\#Marginals}
& \textbf{Time (s)} \\
\midrule
\textbf{Discrete} &
\texttt{\aim}
& 0.0035
& 0.0224
& 0.0556
& 0.0451
& 40.33
& \textbf{399.2} \\
&
\texttt{\aimgdp} (ours)
& \textbf{0.0033}
& \textbf{0.0216}
& \textbf{0.0542}
& \textbf{0.0439}
& \textbf{40.68}
& 404.6 \\
\midrule
\textbf{End-to-end} &
\texttt{\aim}
& 0.0038
& 0.0284
& 0.0763
& 0.0613
& 42.20
& \textbf{426.3} \\
&
\texttt{\aimgdp} (ours)
& \textbf{0.0036}
& \textbf{0.0274}
& \textbf{0.0746}
& \textbf{0.0599}
& \textbf{42.38}
& 444.8 \\
\bottomrule

\end{tabular}
}
\end{table}

\paragraph{Utility.}
\aimgdp consistently outperforms the standard \aim analyzed with zCDP, reducing marginal error by 3\% on average across datasets/privacy budgets.
See top of \cref{tab:aimgdp_utility} for a summary. 
In the Appendix, we present detailed comparisons (\cref{fig:aimgdp_full}), and additional diagnostics over 25 runs on the Adult dataset with $\varepsilon=1$, which show that \aimgdp adds less noise per measurement step, enabling it to run for more iterations (\cref{fig:aimgdp_diagnostics}), while selecting a very similar set of marginals (\cref{fig:aimgdp_network}).
Thus, our tighter accounting improves utility without substantially changing \aim's underlying behavior.
These results provide evidence that \textbf{\aimgdp can serve as a drop-in replacement for standard \aim}.

\paragraph{Attack-aware calibration.}
Choosing $\varepsilon$ is challenging in practice because it does not directly quantify an operational privacy risk.
In contrast, the $\mu$-GDP guarantee of \aimgdp, and the PLD representation of \aimgdp after composition with data preprocessing, directly characterize~\cite{kulynych2024attack} the full trade-off of any membership inference attacks (MIAs)~\cite[see, e.g.,][]{shokri2017membership,carlini2022membership}, which is a standard way to quantify the operational privacy risk.
This enables practitioners to specify an operational risk target, e.g., a maximum attack true positive rate (TPR) at a given low false positive rate (FPR), and calibrate \aimgdp directly to it, rather than first selecting $(\varepsilon,\delta)$ and subsequently translating it to attack risk~\cite{kulynych2024attack}. See \cref{app:riskcal} for details. In \cref{fig:attack} we show the results of calibrating \aimgdp to attack advantage~\cite{yeom2018privacy}, a standard risk measure that equals maximum TPR $-$ FPR of MIAs. As the plot shows, attack-aware calibration substantially improves the utility at the same level of risk (e.g., by 51\% on Adult).
To the best of our knowledge, this is the first such application to DP synthetic data.

\paragraph{Auditing.} It is crucial to audit~\cite[see, e.g.][]{nasr2023tight,annamalai2024what,ganev2025elusive} practical implementations of DP algorithms, as they often introduce software bugs that invalidate the theoretical guarantees. Our GDP analysis enables us to tightly audit the privacy guarantees of \aimgdp. 
To do so, we train 10,000 models on worst-case neighboring datasets and use the FPR/TPR of a MIA to statistically estimate a lower bound on the empirical privacy parameter $\mu_{\mathrm{emp}}$. See \cref{app:audit} for details.
As shown in \cref{fig:audit}, our attacks achieve $\mu_{\mathrm{emp}}=0.50$, closely approaching the theoretical bound of $\mu=0.53$.
To the best of our knowledge, this is the first tight end-to-end empirical audit of fully adaptive \aim, i.e., including its data-dependent marginal selection, as prior work restricted \aim to fixed marginal structures (e.g., all 1-way marginals) and did not establish comparable tightness~\cite{ganev2026tight}.

\subsection{End-to-end synthetic data with \aimgdp, DP preprocessing, and \mixcal}
\label{sec:exp-aimgdp-e2e}

We now evaluate \aimgdp in an end-to-end synthetic data pipeline, following the evaluation setup from the previous section, but privately preprocessing every column via either DP quantiles or Threshold Gaussian, as described in \cref{sec:applications}. We use \mixcal with zCDP as ``dummy'' budget to calibrate the pipeline components, allocating 10\% for preprocessing and 90\% for synthetization.

As before, \aimgdp outperforms \aim (with RDP accounting in this setting), reducing marginal error by over 2\% across datasets
(see bottom of \cref{tab:aimgdp_utility} and \cref{fig:aimgdp_full_e2e} in the Appendix).
Compared with the discrete setting in \cref{sec:exp-aimgdp}, errors decrease for predominantly categorical datasets but increase for predominantly numerical ones, showing that there are avenues for improvement in numerical preprocessing mechanisms.

Next, we illustrate how our end-to-end analyses can be used to support large-scale settings (e.g., more than 200K records) using the Poisson subsampling results (see \cref{sec:technical-results}). For this, we simulate a workload with 16 numerical attributes, each discretized into 16 bins using 4 iterations of DP Quantiles.
The resulting pipeline composes 64 Exponential mechanisms. We compare our tight PLD accounting, which applies subsampling directly in the PLD space, against an RDP baseline that bounds the Exponential mechanisms via zCDP. Our analysis yields a lower $\varepsilon$ across all sampling rates, improving over RDP by 5.6\% without subsampling ($q=1$) and by \emph{up to 66.7\%} at $q=2^{-6}$, with the gains increasing as the subsampling probability decreases. We show the results in \cref{fig:subsampled_accounting} in the Appendix.

\subsection{Generating Synthetic Text Data with LLMs and Tabular Conditioning}
To further demonstrate the benefits our methods, we replicate the attribute-conditioned text generation (ACTG) pipeline~\cite{hu2026actg}. 
This work has improved on existing DP synthetic text evaluations~\cite{xie2024differentially, tan2025synthesizing}, and found that ACTG outperformed strong baselines such as vanilla \dpsgd, private evolution~\cite{xie2024differentially}, and CTCL~\cite{tan2025synthesizing}. 
ACTG assumes records that consist of structured tabular attributes and text. ACTG uses \aim to produce a synthetic dataset of structured attributes, then finetunes a language model with \dpsgd to generate text conditioned on these attributes (see \cref{fig:d_pipeline}). We reproduce this pipeline on the \emph{bioRxiv} dataset using the Gemma 3 1B base model. ACTG has a reinforcement learning stage, but, as it is independent of the privacy analysis, we omit it in our demonstration. We compare the original \aim vs. \aimgdp with end-to-end PLD accounting using \mixcal. 
We used budget allocation weights that replicate the allocation used by \citet{hu2026actg}. See further details in \cref{app:extra-exp}.

We report our results in terms of the standard MAUVE scores~\cite{pillutla2023mauve} in \cref{fig:utility}. We observe consistent improvement in using PLD with \mixcal and \aimgdp, over RDP accounting with \aim. This gain is due to a combination of using \aimgdp over \aim and using lower noise parameters for \dpsgd when using PLD accounting compared to RDP accounting. Additionally, we show the improvements in noise multipliers of \dpsgd in \cref{tab:actg_noise_multipliers}, using Poisson subsampling and Random Allocation, which is a new capability enabled by PLD accounting compared to the RDP approach.

\begin{table}[t]
\footnotesize
\centering
\caption{\textbf{Our analyses enable adding substantially less noise in ACTG pipelines (combination of synthetization and \dpsgd).} The table shows \dpsgd noise multipliers $\sigma_{\mathrm{SGD}}$ 
at fixed iterations $T = 512, K = 8$, $\delta = 10^{-6}$ for varying values of $\varepsilon$. Poisson subsampling uses $q = K/T$.}
\label{tab:actg_noise_multipliers}
\resizebox{0.55\linewidth}{!}{
\begin{tabular}{cccc}
\toprule
$\varepsilon$ & \textbf{RDP (Poisson)} & \textbf{PLD (Poisson)} & \textbf{PLD (Random Alloc.)} \\
\midrule
1 & 2.34 & \textbf{2.18} & 2.25 \\
2 & 1.40 & 1.32 & \textbf{1.30} \\
4 & 0.97 & 0.92 & \textbf{0.89} \\
8 & 0.74 & 0.71 & \textbf{0.69} \\
\bottomrule
\end{tabular}
}
\vspace{-1em}
\end{table}

\section{Conclusions}
We unified disparate DP primitives under a single analytical and algorithmic framework, providing both qualitative and quantitative advantages for DP synthetic data pipelines with heterogeneous components. To support this, we presented GDP and PLD analyses for bounded range mechanisms, alongside a series of other results such as the PLD analysis of private partition selection (e.g., Threshold Gaussian). These analyses provide a versatile toolbox for future DP mechanism design. Applying these theoretical results, we introduced \mixcal and \aimgdp to calibrate and tightly account for mixed-mechanism workloads. By reducing the analytical slack in heterogeneous composition, our framework yields immediate utility gains across tabular and multimodal use cases without requiring major mechanism modifications. Ultimately, this work establishes a principled foundation for building and tightly accounting for increasingly complex, multimodal DP mechanisms.

\paragraph{Limitations.} Our GDP and PLD analyses of BR mechanisms do not provide tighter accounting than standard zCDP in \textit{all} cases. E.g., as we show in \cref{sec:technical-results}, GDP is only beneficial in heterogeneous pipelines when the majority of the budget is used for the Gaussian mechanisms, which is the case for \aim, but not other fully adaptive workflows, e.g., interactive statistical queries. 
Our \mixcal algorithm enables allocating the budget according to pre-specified weights, but we do not introduce a way for optimally choosing these weights. Finding optimal ways to do so is an avenue for future work.
Finally, although the GDP analysis used in \aimgdp is more precise than the zCDP analysis, it is not as tight as it would have been had the algorithm used only homogeneous components internally. Future work could compare our version to possible variants of \aim that, e.g., only use Gaussian mechanisms, which would enable even tighter analyses.

\subsection*{AI use statement}
We have used AI-based systems to assist with the implementation of the code for the experiments, visualizations, and LaTeX source code, as well as with proof ideation and verification. All the final text and proofs are written by the authors.

\bibliography{main}
\bibliographystyle{iclr2027_conference}

\clearpage
\begin{appendices}
\section{Proofs and Analysis of Bounded Range Mechanisms}
\label{app:proofs}

\subsection{BR to $f$-DP, privacy profiles and PLD}

The proof of the theorem rests on one result, which \citet{dong2020optimal} establish in the
course of proving their reduction to generalized randomized response (their
Lemma 4.1). 
\begin{lemma}[\citealp{dong2020optimal}]\label{stmt:br-to-beta}
    Suppose that a mechanism $M: 2^\sD \to \sY$ satisfies $\nu$-BR. Then, for any $S \simeq S'$ there exists $0 \leq t \leq \nu$ such that for any test $\phi: \sY \rightarrow \{0, 1\}$, we have:
    \begin{equation}\label{eq:br-to-bound}
        \beta_\phi \geq \max\{1 - e^t \alpha_\phi,\ e^{t - \nu}(1 - \alpha_\phi)\}.
    \end{equation}

    Moreover for each $\alpha_\phi\in[0, 1]$, there exists the mechanism, such that 
    the inequality turns to the equality.
\end{lemma}

We show the result next.
\brprivprof*

\begin{proof}

We prove the statements in three main steps.

\paragraph{The $f$-DP trade-off function and its tightness.}

    To show the result, for a given $\alpha$ it suffices to find the smallest \cref{stmt:br-to-beta} bound for all $t \in [0, \nu]$:
    \begin{equation}\label{eq:proof-br-f-obj}
        f(\alpha) = \min_{t \in [0, \nu]} \max \{1 - e^t \alpha,~ e^{t- \nu} (1 - \alpha)\}.
    \end{equation}
    As the first term is decreasing in $t$, and the second one is increasing, the optimal value must occur at their intersection. Equating the two terms and solving for $t$  we get:
    \begin{equation}
        1 - e^t \alpha = e^{t - \nu} (1 - \alpha) \Rightarrow e^{t} = \frac{1}{\alpha + (1 - \alpha) e^{-\nu}}.
    \end{equation}
    Substituting $e^t$ into \cref{eq:proof-br-f-obj}, we have:
    \begin{equation}
        f(\alpha) = 1 - \frac{\alpha}{\alpha + (1 - \alpha) e^{-\nu}} = \frac{(1 - \alpha) e^{-\nu}}{\alpha + (1 - \alpha) e^{-\nu}} = \frac{1 - \alpha}{\alpha e^\nu + 1 - \alpha}.
    \end{equation}

    The tightness follows from the fact that each point $(\alpha, f(\alpha))$ is attainable by \cref{stmt:br-to-beta}.

\paragraph{From $f$-DP to privacy profile.}

Let $a = e^\nu - 1$, so $f(\alpha) = \frac{1-\alpha}{1+a\alpha}$, and write

$$\delta(\varepsilon) = \sup_{\alpha\in[0,1]} g(\alpha), \text{ where } g(\alpha) = 1 - e^\varepsilon\alpha - f(\alpha).$$

\textit{Concavity}. As $f'(\alpha) = -\frac{e^\nu}{(1+a\alpha)^2}$ and $f'' > 0$, the function $g$ is concave. Its derivative $g'(\alpha) = \frac{e^\nu}{(1+a\alpha)^2} - e^\varepsilon$ decreases from $g'(0) = e^\nu - e^\varepsilon$ to $g'(1) = e^{-\nu} - e^\varepsilon$.

\textit{Boundary cases.}

If $\varepsilon \ge \nu$, then $g' \le 0$ throughout, so $\delta = g(0) = 0$.
If $\varepsilon \le -\nu$, then $g' \ge 0$ throughout, so $\delta = g(1) = 1 - e^\varepsilon$.

\textit{Interior case: $|\varepsilon| < \nu$}. In this case, $g'(\alpha^*) = 0$ at $1 + a\alpha^* = s \define e^{(\nu-\varepsilon)/2} \in (1, e^\nu)$. At this point, $f(\alpha^*) = \frac{e^\nu - s}{as}$, and $e^\varepsilon s = \frac{e^\nu}{s} = e^{(\nu+\varepsilon)/2}$. Therefore:

$$\delta(\varepsilon) = \frac{a - e^\varepsilon(s-1) - (e^\nu - s)/s}{a} = \frac{e^\nu + e^\varepsilon - 2e^{(\nu+\varepsilon)/2}}{e^\nu - 1} = \frac{\big(e^{\nu/2} - e^{\varepsilon/2}\big)^2}{e^\nu - 1}.$$

At $\varepsilon = \pm\nu$ this equals $0$ and $1 - e^{-\nu}$ respectively, which matches the boundary cases.
The tightness follows from the tightness of the tradeoff function.

\paragraph{From the privacy profile to PLD.}

If $\delta(\varepsilon)$ is twice-differentiable profile function, then by \citet{sommer2019privacy}, we have that $g = \delta''-\delta'$ is the density of PLD.

On $(-\nu,\nu)$, differentiating
$\delta(\varepsilon) = \big(e^{\nu} - 2e^{(\nu+\varepsilon)/2}
+ e^{\varepsilon}\big)/(e^{\nu}-1)$ gives:
\begin{equation}
    \delta_\nu'' - \delta_\nu'
    = \frac{\left(e^{\varepsilon}-\tfrac12e^{(\nu+\varepsilon)/2}\right)
          - \left(e^{\varepsilon}-e^{(\nu+\varepsilon)/2}\right)}
           {e^{\nu}-1}
    = \frac{e^{\nu/2}}{2\left(e^{\nu}-1\right)}e^{\varepsilon/2}
    = D_\nu\,e^{\varepsilon/2},
\end{equation}
where $2\big(e^{\nu/2}-e^{-\nu/2}\big)
= 2e^{-\nu/2}\big(e^{\nu}-1\big)$ gives
\begin{equation}\label{eq:Dnu}
    D_\nu = \frac{1}{2\left(e^{\nu/2}-e^{-\nu/2}\right)}
          = \frac{e^{\nu/2}}{2\left(e^{\nu}-1\right)} .
\end{equation}
Observe that outside of the interval $[-\nu,\nu]$, the density vanishes: $\delta(\varepsilon) \equiv 0$ on $(\nu,\infty)$,
and on $(-\infty,-\nu)$ we have
$\delta_\nu''-\delta_\nu' = -e^{\varepsilon}+e^{\varepsilon} = 0$.

\end{proof}

Next, we provide several observations about these results.

\paragraph{Dominating PLD of a $\nu$-BR mechanism corresponds to $2\nu$-BR.}
The dominating PLD $g_\nu$ in \cref{stmt:br-to-pld} is supported on $[-\nu, \nu]$, so the
corresponding mechanism is, surpisingly, $2\nu$-BR and not $\nu$-BR, as one would expect. This is because a pair dominating the whole class, for all neighbouring datasets, must
have privacy loss reaching up to $\nu$, and, applying the same argument to the
reversed pair, reaching down to $-\nu$. Consequently, no dominating pair with privacy loss having bounded range $\nu$, and thus no single
$\nu$-BR mechanism, dominates the entire class of BR mechanisms: every point of the bounds in
\cref{stmt:br-to-pld} is attained by some $\nu$-BR mechanism, yet no single mechanism attains
all of them simultaneously. This contrasts with, e.g., the Laplace and Gaussian mechanisms, including
their subsampled versions, whose tightly dominating pairs are realized by the mechanism itself
on a worst-case pair of inputs.

\paragraph{Suboptimality of the PLD analysis.}
Theorem 3 of \citet{dong2020optimal} provides the exact optimal privacy profile for the
homogeneous non-adaptive composition of $k$ mechanisms that are $\nu$-BR. 
At $k=1$, the two results coincide. For any $k>1$, however, the PLD approach of the $k$-fold convolution of $g_\nu$ should be larger (i.e., worse guarantees), as the result from \citeauthor{{dong2020optimal}} is optimal for the homogeneous composition of BR mechanisms.

At the same time, the PLD analysis has an important advantages over the optimal expression: It enables simple composition with any mechanisms that also support PLD analyses, unlike the optimal expression that only works for homogenous compositions of BR mechanisms. The looseness is thus the price of that flexibility.

\paragraph{Despite the suboptimality, PLD analysis is tighter than zCDP in low-composition regimes.}
In \cref{fig:br-to-pld} we show empirical evaluations illustrating that, under homogeneous composition of BR mechanisms, for some number of compositions, zCDP provides a better $(\varepsilon, \delta)$-DP bound. The explanation is that both are upper bounds and the tight PLD of BR is not bounded-range itself, thus it might have larger zCDP.

\begin{figure}
    \centering
    \resizebox{\linewidth}{!}{\input{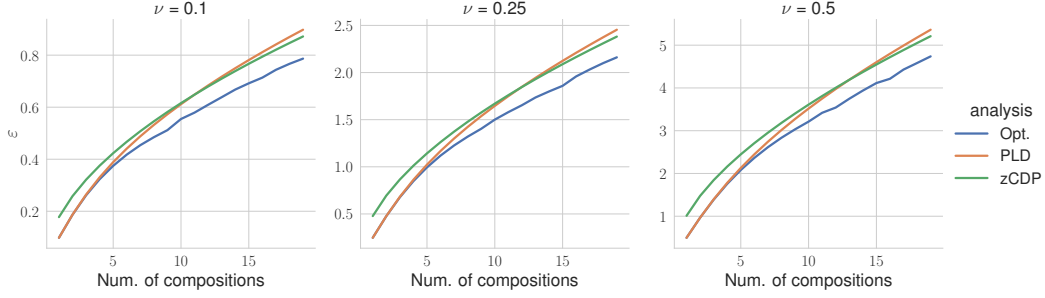}}
    \caption{The DP $\varepsilon$ at $\delta = 10^{-5}$ reported using the PLD (ours), zCDP~\cite{harrison2025exact}, and optimal~\cite{dong2020optimal} analyses for the compositions of homogeneous BR mechanisms.}
    \label{fig:br-to-pld}
\end{figure}

\paragraph{Simple form of a TV privacy guarantee.}
We observe an interesting form in terms of \emph{total-variation (TV) privacy}~\cite[see, e.g.,][]{geng2015optimal,barber2014privacy,kulynych2022what}:
An algorithm satisfies $\eta$-TV privacy if it satisfies $(0, \eta)$-DP.
TV is not a useful privacy notion on its own, but it provides simple forms of risk and generalization bounds~\cite{cherubin2024closed,kulynych2025unifying}.
For the classical notion of $(\varepsilon, 0)$-DP, \citet{kairouz2015composition} showed that an $(\varepsilon, 0)$-DP mechanism satisfies $\eta$-TV privacy with:
\begin{equation}\label{eq:dp-to-tv}
    \eta = \frac{e^{\varepsilon} - 1}{e^{\varepsilon} + 1} = \mathrm{tanh}(\nicefrac{\varepsilon}{2}).
\end{equation}

The conversion from BR to TV privacy is as follows:
\begin{corollary}
    Suppose that a mechanism $M(\cdot)$ satisfies $\nu$-BR. Then, it satisfies $\eta$-TV privacy with:
    $$\eta = \frac{e^{\frac{\nu}{2}} - 1}{e^{\frac{\nu}{2}} + 1} = \mathrm{tanh}(\nicefrac{\nu}{4}).$$
\end{corollary}

Curiously, this is exactly the conversion from $\varepsilon$-DP to TV privacy from \cref{eq:dp-to-tv} with a difference of a factor of two.

\subsection{PLD of Laplace mechanism as a mixture of RR and BR}

Next, we discuss an interesting relationship of the result in \cref{stmt:br-to-pld} to the PLD of the Laplace mechanism.
We write $\mathrm{Lap}_\nu$ for the Laplace mechanism of scale $1/\nu$ on a
statistic of sensitivity $1$, which is $\nu$-DP, and let
$\omega^{\mathrm{Lap}}_\nu$ denote its PLD. Define the following distribution:
\begin{equation}
    \omega^{\mathrm{RR}}_\nu :=
    \begin{cases}
        \nu, & \text{w.p. } \dfrac{e^{\nu}}{1+e^{\nu}}, \\[6pt]
        -\nu, & \text{w.p. } \dfrac{1}{1+e^{\nu}},
    \end{cases}
\end{equation}
which is the PLD of the standard randomized response (RR) mechanism with parameter $\nu$, which tightly dominates
the class of $\nu$-DP mechanisms \cite{kairouz2015composition}, and let
$\omega^{\mathrm{BR}}_\nu$ be the PLD with density $g_\nu$ of Theorem
\ref{stmt:br-to-pld}. All three are supported on $[-\nu,\nu]$.

The Laplace PLD is well-known:
\begin{lemma}[{\citealp[App.~A.3]{sommer2019privacy}}]\label{lem:laplace-pld}
$\omega^{\mathrm{Lap}}_\nu$ is the measure on $[-\nu,\nu]$ consisting of two
atoms,
\begin{equation}
    \Pr\left[\omega^{\mathrm{Lap}}_\nu = \nu\right] = \tfrac12,
    \qquad
    \Pr\left[\omega^{\mathrm{Lap}}_\nu = -\nu\right] = \tfrac12 e^{-\nu},
\end{equation}
together with an absolutely continuous part on $(-\nu,\nu)$ of density
\begin{equation}
    \tfrac14\,e^{-\nu/2}\,e^{w/2} .
\end{equation}
\end{lemma}

Next, we show the that the PLD of Laplace mechanism is a mixture of the PLD of randomized response and bounded range:
\begin{proposition}\label{prop:laplace-mixture}
For every $\nu > 0$,
\begin{equation}\label{eq:mixture}
    \omega^{\mathrm{Lap}}_\nu
    = \lambda_\nu\, \omega^{\mathrm{RR}}_\nu
      + \left(1-\lambda_\nu\right) \omega^{\mathrm{BR}}_\nu ,
    \qquad
    \lambda_\nu = \frac{1+e^{-\nu}}{2},
\end{equation}
as an identity of distributions on $\sR$.
\end{proposition}

\begin{proof}
By Lemma \ref{lem:laplace-pld} the atoms carry total mass
$\tfrac12\left(1+e^{-\nu}\right) = \lambda_\nu$, and after
normalisation their masses are those of $\omega^{\mathrm{RR}}_\nu$, since
$\frac{1/2}{\lambda_\nu} = \frac{1}{1+e^{-\nu}}
= \frac{e^{\nu}}{1+e^{\nu}}$. The continuous part carries the
remaining mass:
\begin{equation}
    \tfrac14e^{-\nu/2}\!\int_{-\nu}^{\nu}\!e^{w/2}\mathrm{d}w
    = \tfrac12\left(1-e^{-\nu}\right) = 1 - \lambda_\nu ,
\end{equation}
and normalising its density gives
$\frac{\tfrac14e^{-\nu/2}}{1-\lambda_\nu}e^{w/2}
= \frac{e^{w/2}}{2\left(e^{\nu/2}-e^{-\nu/2}\right)}
= D_\nu e^{w/2}$, i.e.\ exactly $g_\nu$.
\end{proof}

Because $\delta(\varepsilon)
= \mathbb{E}\big[(1-e^{\varepsilon-\omega})_{+}\big]$ is \emph{linear}
in the law of $\omega$, the mixture form of \eqref{eq:mixture} also extends to privacy profiles.

\begin{corollary}\label{cor:laplace}
For every $\varepsilon \in \sR$,
\begin{equation}\label{eq:delta-mixture}
    \delta_{\mathrm{Lap}_\nu}(\varepsilon)
    = \lambda_\nu\,\delta_{\omega^{\mathrm{RR}}_\nu}(\varepsilon)
      + \left(1-\lambda_\nu\right)\delta_\nu(\varepsilon)
    \;\ge\; \delta_\nu(\varepsilon),
\end{equation}
that is, $\omega^{\mathrm{BR}}_\nu \preceq \omega^{\mathrm{Lap}}_\nu$. In
particular the Laplace PLD is itself a valid, but strictly loose, dominating PLD
for the $\nu$-BR class.
\end{corollary}

\begin{proof}
The equality is linearity of $\delta$ in the mixture. Since $\omega^{\mathrm{RR}}_\nu$ dominates every $\nu$-DP mechanism
\cite{kairouz2015composition}, we have $\delta_{\omega^{\mathrm{RR}}_\nu} \ge \delta_\nu$
pointwise. Substituting this into the equality and using
$\lambda_\nu + (1-\lambda_\nu) = 1$ gives the inequality.
\end{proof}

\subsection{BR to GDP}

In this section, we prove the GDP guarantee for BR mechanisms:
\brtogdp*

First, we make use of the properties of the \emph{inverse Mills ratio}.

\begin{proposition}[\citealp{sampford1953some}]\label{prop:mills-convex}
The inverse Mills ratio $\lambda(t) \define \nicefrac{\varphi(t)}{\Phi(t)}$, where $\varphi(t)$ is the density of the standard Normal distribution, is convex on $\sR$.
In particular, $\lambda'$ is nondecreasing.
\end{proposition}

We express the statement in \cref{stmt:br-to-gdp} as a condition in terms of an expression $\log \left(\nicefrac{\Phi(\nicefrac{t}{2})}{\Phi(-\nicefrac{t}{2})}\right)$, which has the following properties:
\begin{restatable}{lemma}{lderivative}\label{stmt:L-derivative}
Let $L(t) \define \log \frac{\Phi(\nicefrac{t}{2})}{\Phi(-\nicefrac{t}{2})}$. Then, $L$ is odd and differentiable. Its derivative $L'$ is even, and nondecreasing on $[0, \infty)$.
\end{restatable}

\begin{proof}
Oddness is immediate from $L(-t) = \log\frac{\Phi(-t/2)}{\Phi(t/2)} = -L(t)$. Differentiating, we get:
\[
    L'(t) = \frac{\dd}{\dd t}\log\Phi\left(\frac{t}{2}\right) - \frac{\dd}{\dd t}\log\Phi\left(-\frac{t}{2}\right)
    = \frac{\varphi(\nicefrac{t}{2})}{2\,\Phi(\nicefrac{t}{2})} + \frac{\varphi(-\nicefrac{t}{2})}{2\,\Phi(-\nicefrac{t}{2})}
    = \frac{1}{2}\left[\lambda\left(\frac{t}{2}\right) + \lambda\left(-\frac{t}{2}\right)\right],
\]
using $\varphi(-s) = \varphi(s)$. The derivative $L'$ is even as $L$ is odd. Finally, for any $t \geq 0$,
\[
    \frac{\dd}{\dd t} L'(t) = \frac{1}{4}\left[\lambda'\left(\frac{t}{2}\right) - \lambda'\left(-\frac{t}{2}\right)\right] \geq 0,
\]
as $\lambda'$ is nondecreasing by \cref{prop:mills-convex} and $\nicefrac{t}{2} \geq -\nicefrac{t}{2}$ for $t \geq 0$.
\end{proof}

Finally, we can show the proof of \cref{stmt:br-to-gdp}.

\begin{proof}
To prove the statement, we need to show that we have, for any $\alpha \in [0, 1]$:
\begin{equation}\label{eq:proof-gdp-condition}
    \Phi(\Phi^{-1}(1 - \alpha) - \mu) \leq \frac{1 - \alpha}{\alpha e^\nu + 1 - \alpha},
\end{equation}
where $\mu$ is defined as in \cref{eq:br-to-gdp}. To accomplish this, we will simplify this condition to a form that is more amenable to analysis.
For this, first, let us write down $\nu$ in terms of $\mu$ by solving \cref{eq:br-to-gdp}:
\[
    \Phi\left(-\frac{\mu}{2}\right) = \frac{1}{e^{\frac{\nu}{2}} + 1}.
\]
As $\Phi(t) = 1 - \Phi(-t)$, we have that:
\[
    \Phi\left(\frac{\mu}{2}\right) = 1 - \frac{1}{e^{\frac{\nu}{2}} + 1} = \frac{e^{\frac{\nu}{2}}}{e^{\frac{\nu}{2}} + 1}.
\]
Dividing the two expressions, we have:
\[
    \frac{\Phi(\nicefrac{\mu}{2})}{\Phi(-\nicefrac{\mu}{2})} = e^{\frac{\nu}{2}}.
\]
Using the definition of $L(t)$ in \cref{stmt:L-derivative}, we thus have:
\begin{equation}\label{eq:proof-gdp-to-br}
    \nu = 2L\left(\mu\right).
\end{equation}
Next, denoting by $z = \Phi^{-1}(1 - \alpha)$, so that $\alpha = 1 - \Phi(z) = \Phi(-z)$, we have:
\begin{equation}\label{eq:proof-intermediate-frac}
    \Phi(z - \mu) \leq \frac{\Phi(z)}{\Phi(-z) e^\nu + \Phi(z)}.
\end{equation}
Observing that the denominator in r.h.s. is always positive, we can multiply both sides:
\[
    \Phi(z - \mu) \Phi(-z) e^\nu + \Phi(z - \mu) \Phi(z) \leq \Phi(z).
\]
Rewriting to group the $\Phi(z)$ terms:
\begin{align*}
    \Phi(z - \mu) \Phi(-z) e^\nu \leq \Phi(z) \left( 1 - \Phi(z - \mu)\right) = \Phi(z) \Phi(\mu - z).
\end{align*}
Taking logarithms on both sides:
\begin{align*}
    \log \Phi(z - \mu) + \log \Phi(-z) + \nu \leq &\log \Phi(z) + \log \Phi(\mu - z) \\
    \nu + \log \frac{\Phi(z - \mu)}{\Phi(\mu - z)} &\leq \log \frac{\Phi(z)}{\Phi(-z)}.
\end{align*}
Using the definition of $L(t)$ above and the form of $\nu$ in \cref{eq:proof-gdp-to-br}:
\[
    2L\left(\mu\right) \leq L(2z) - L(2z - 2\mu).
\]
Next, to derive a useful form, we write $z = \frac{\mu}{2} + \frac{u}{2}$ for some $u \in \sR$:
\[
    2L\left(\mu\right) \leq L\left(u + \mu\right) - L(u - \mu).
\]

As $L$ is odd by \cref{stmt:L-derivative}, we have $-L(u - \mu) = L(\mu - u)$. Therefore, the condition in \cref{eq:proof-gdp-condition} that we need to show is equivalent to the following:
\begin{equation}\label{eq:proof-gdp-midpoint}
    D(u) \define L(\mu + u) + L(\mu - u) - 2 L(\mu) \geq 0 \qquad \forall u \in \sR.
\end{equation}

This form is easier to analyze. We start by observing that, as $D$ is even, we have $D(u) = D(|u|)$ for any $u \in \sR$, thus it suffices to consider
$u \geq 0$. Moreover, observing that $D(0) = 0$, it suffices to show that $D$ is nondecreasing on $[0, \infty)$. Its derivative is:
\[
    D'(u) = L'(\mu + u) - L'(\mu - u).
\]
By \cref{stmt:L-derivative}, $L'$ is even and nondecreasing on $[0, \infty)$. This implies that $L'(a) \geq L'(b)$ whenever $|a| \geq |b|$. As
$|\mu + u| \geq |\mu - u|$ for $\mu, u \geq 0$, we get $D'(u) \geq 0$, and the
claim follows.
\end{proof}

\section{Additional Accounting Results}
\label{app:extra-methods}

\subsection{Poisson Subsampling for Arbitrary Privacy Loss Distributions}
\label{sec:tech_poisson}

We present a novel method for converting a privacy guarantee for an arbitrary mechanism to a privacy guarantee for a Poisson subsampled version of that mechanism, which is compatible with modern privacy accounting libraries such as \texttt{dp\_accounting}. 

Our method operates by representing the privacy guarantee as a privacy profile.
Although we could then convert the privacy profile to a PLD representation, we found empirically that in this case the PLD approach is less numerically stable to apply the conversion to than the $\delta(\varepsilon)$ representation. 

We recall the hockey-stick divergence has a form $H_\alpha(P, Q) = \int (P(x) - \alpha Q(x))_+ dx$, and
if a dominating pair $(P, Q)$ for $M(\cdot)$ satisfies $H_{e^\varepsilon}(P, Q) \leq \delta$ then $M(\cdot)$ is $(\varepsilon, \delta)$-DP. Finally, Theorem 11 of \citet{zhu2022optimal} shows that if $(P, Q)$ is a dominating pair for the unsampled mechanism in the remove direction, then $(P', Q)$, where $P' = q \cdot P + (1 - q) \cdot Q$, is a dominating pair for Poisson-subsampled mechanism with sampling probability $q$ in the remove direction, and vice-versa in the add direction. 

The following theorem gives a conversion from hockey-stick divergences between $P, Q$ to hockey-stick divergences between $P', Q$, which by the above facts suffices to convert a privacy profile for a non-subsampled mechanism to a privacy profile for its Poisson-subsampled counterpart.

\begin{theorem}\label{thm:poissonarbitrary}
Let $P' = q \cdot P + (1 - q) \cdot Q$. Then:
\[H_{e^\varepsilon}(P' , Q) = 
\begin{cases}
    q \cdot H_{\frac{e^\varepsilon + q - 1}{q}}(P , Q)&\text{ if } e^\varepsilon > 1-q,\\1-e^\varepsilon&\text{ otherwise}.
\end{cases}\]
\[H_{e^\varepsilon}(Q , P') = 
\begin{cases}
    (1 - e^\varepsilon(1-q)) \cdot H_{\frac{e^\varepsilon \cdot q}{1 - e^\varepsilon(1 - q)}}(Q , P)&\text{ if } e^\varepsilon < \frac{1}{1 - q},\\0&\text{ otherwise}.
\end{cases}\]
\end{theorem}
\begin{proof}
    We have:
    \[
        H_{e^\varepsilon}(P' , Q) = \int (P'(y) - e^\varepsilon Q(y))_+ dy = \int (q \cdot P(y) + (1 - q - e^\varepsilon) Q(y))_+ dy. \\
    \]
    If $e^\varepsilon \leq 1 - q$ then $1 - q - e^\varepsilon \geq 0$ and:
    \begin{align*}
        &\int (q \cdot P(y) + (1 - q - e^\varepsilon) Q(y))_+ dy \\
        =& \int q \cdot P(y) + (1 - q - e^\varepsilon) Q(y) dy \\
        =& q \int P(y) dy + (1 - q - e^\varepsilon) \int Q(y) dy \\
        =& q + 1 - q - e^\varepsilon = 1 - e^\varepsilon.
    \end{align*}
    Otherwise: 
    \begin{align*}
        H_{e^\varepsilon}(P' , Q) &= \int (q \cdot P(y) - (e^\varepsilon + q - 1) Q(y))_+ dy \\
        &= q \int (P(y) - \frac{e^\varepsilon + q - 1}{q} Q(y))_+ dy = q \cdot H_{\frac{e^\varepsilon + q - 1}{q}}(P , Q),
    \end{align*}
    Similarly we have:
    \[
        H_{e^\varepsilon}(Q , P') = \int (Q(y) - e^\varepsilon P'(y))_+ dy = \int ((1 - e^\varepsilon(1 - q)) \cdot Q(y) - e^\varepsilon q \cdot P(y))_+ dy.
    \]
    If $e^\varepsilon < \frac{1}{1 - q}$ then:
    \begin{align*}
        &\int ((1 - e^\varepsilon(1 - q)) \cdot Q(y) - e^\varepsilon q \cdot P(y))_+ dy. \\
        =& (1 - e^\varepsilon(1 - q)) \int (Q(y) - \frac{e^\varepsilon q}{1 - e^\varepsilon(1 - q)} \cdot P(y))_+ dy \\
        =& (1 - e^\varepsilon(1 - q)) H_{\frac{e^\varepsilon q}{1 - e^\varepsilon(1 - q)}}(Q , P),
    \end{align*}
    Otherwise, $1 - e^\varepsilon(1 - q) \leq 0$ and so:
    \[\int ((1 - e^\varepsilon(1 - q)) \cdot Q(y) - e^\varepsilon q \cdot P(y))_+ dy = \int 0\ dy = 0.\]
\end{proof}

\subsection{Exact Characterization of Private Partition Selection}\label{sec:privateps}

In this section, we present characterizations of mechanisms such as the Threshold Gaussian and Threshold Laplace~\cite{gopi20setunion} for private partition selection. It rests on the following result:
 
\begin{restatable}{lemma}{deltacomp}\label{stmt:delta-comp}
    Let $M$ be a mechanism with the following property: There exists $P, Q$ such that for any inputs $S \simeq S'$ where $S$ contains an example not in $S'$, there is an event $E$ occurring with probability at least $1 - \tau$ such that:
    \[T(M(S)|E, M(S'))(\alpha) \geq T(P, Q)(\alpha)\] for all $\alpha \in [0, 1]$, where $A|E$ denotes the random variable $A$ conditioned on event $E$. Then $P \times \bern(\tau), Q \times \bern(0)$ is a dominating pair of $M$, where $\times$ denotes a product distribution. 
\end{restatable}

\begin{proof}
    We use the following fact, a corollary of Blackwell's informativeness theorem \cite{blackwell1953equivalent}: $(P, Q)$ is a dominating pair for $M(\cdot)$ if and only if for any $S \simeq S'$, there is a randomized function $f$ such that $T(P) \equiv M(S), T(Q) \equiv M(S')$ where $\equiv$ denotes distributional equivalence.
   
   Fix a $S \simeq S'$, and let $T(\cdot)$ be the randomized function such that 
   $M(S)|E \equiv T(P), M(S') \equiv T(Q)$. By the preceding fact, it suffices to define a randomized function $T'(\cdot)$ such that $T'(P \times \bern(\tau)) \equiv M(S), T'(Q \times \bern(0)) \equiv M(S')$. Our function is:
    $$
    T'(x, y) = 
    \begin{cases} 
    T(x) & \text{if } y = 0 \\[6pt]
    x' \sim M(S) \mid \bar{E} & \text{if } y = 1 \text{ w.p. } \frac{\Pr[\bar{E}]}{\tau} \\[6pt]
    T(x) & \text{if } y = 1 \text{ w.p. } 1 - \frac{\Pr[\bar{E}]}{\tau}
    \end{cases}
    $$

   We have:
   \begin{align*}
    T'(P \times \bern(\tau)) &\equiv (1 - \tau) \cdot T(P) + \tau \cdot \left(\frac{\Pr[\bar{E}]}{\tau} \cdot M(S) | \bar{E} + \left(1 - \frac{\Pr[\bar{E}]}{\tau}\right) T(P)\right)\\
    &\equiv (1 - \tau + \tau - \Pr[\bar{E}]) \cdot T(P) + \Pr[\bar{E}] \cdot M(S) | \bar{E}\\
    &\equiv \Pr[E] \cdot M(S) | E + \Pr[\bar{E}] \cdot M(S) | \bar{E}\\
   &\equiv M(S).
   \end{align*}

   and more straightforwardly, 
    $T'(Q \times \bern(0)) \equiv T(Q) \equiv M(S')$, completing the proof.
\end{proof}

The proof of \citet[][Theorem 5.1]{gopi20setunion} exactly specifies an event $E$ for their mechanism (namely, $E$ is the event that for $S \simeq S'$ and $S \supset S'$, $M(S)$ is in the support of $M(S')$), which is what implies that their policy Gaussian algorithm (that we refer to as Threshold Gaussian) is a composition of a Gaussian mechanism and a mechanism with dominating pair $(\bern(\tau), \bern(0))$.

\subsection{Group Privacy Guarantees for Basic Preprocessing Mechanisms}
\label{app:group_privacy}

In real-world settings, individual users often contribute multiple records to the dataset.
Practitioners typically handle this by enforcing hard contribution bounding, randomly truncating each user's history to at most one record ($K = 1$).
While conceptually simple, contribution bounding introduces a \emph{truncation bias}.

An alternative is to retain up to $K > 1$ records per user and evaluate the resulting user-level guarantee via group privacy.
A user with $K$ records contributes up to $K \times$ more signal to linear queries, equivalently the sensitivity increases by $K$, requiring $K \times$ more noise to preserve the same signal-to-noise ratio/privacy guarantee.
Below, we formalize how group privacy plumbs through our unified framework across each primitive class.

\begin{itemize}[leftmargin=*]
    \item \textit{Gaussian Mechanism:}
    If a user contributes at most $K$ records, the global $L_2$ sensitivity of a marginal scales by at most $K$.
    Under Gaussian Differential Privacy (GDP), the guarantee scales strictly as $\mu \to K \mu$.
    The dominating pair remains Gaussian ($G_{K\mu}$), retaining exact hypothesis-testing duality.

    \item \textit{Exponential Mechanism:}
    When each record can alter the selection score by at most $\Delta s$, $K$ records alter the score by at most $K \Delta s$, scaling the effective privacy parameter $\varepsilon \to K \varepsilon$.

    \item \textit{Threshold Gaussian:}
    A user with $K$ records can alter the counts of measurements by a total of $K$ across all measurements. After the partition selection primitive, this mechanism is just a Gaussian mechanism and thus we can reduce to the above analysis for group privacy of the Gaussian mechanism.

    For partition selection at cutoff $\theta$, the $(0, \tau)$-DP failure event corresponds to the probability that a measurement with zero count in $S'$ exceeds the threshold $\theta$. In the worst case, all $K$ additional records contribute to measurements with zero count in $S'$. The probability of one of these measurements being released can be formulated as:
    \begin{equation*}
        \tau(K) = \max_{\substack{k_1 \geq \ldots \geq k_j \geq 1 \\ k_1 + \cdots + k_j = K}}
        \left( 1 - \prod_{i = 1}^{j} (1 - p_{k_i}) \right),
    \end{equation*}

    where $p_k$ is the probability that a record with count $k$ is released by the partition selection mechanism:

    \begin{equation*}
        p_k \define = 1 - \Phi\left(\frac{\theta - k}{\sigma}\right).
    \end{equation*}

    When e.g. $K \leq 46$, there are less than $10^5$ distinct terms in the max defining $\tau(K)$ and hence evaluating $\tau(K)$ via brute-force enumeration can be done reasonably efficiently. Calibrating $\theta$ (possibly jointly with $\sigma$) to achieve a desired failure probability $\tau$ can also be done efficiently via binary search.
\end{itemize}

Because every mechanism in the taxonomy admits an exact representation parameterized by $K$, our framework allows us to output precise privacy profiles for multiple group sizes.
With this approach, we can calibrate our mechanism to any group size $K$, then post-hoc report the privacy profiles for any other $K'$ (noting that if there are Threshold Gaussian mechanisms, we must have $K' \leq K$).

One could alternatively consider taking the record-level privacy guarantee of the composed mechanism, and trying to convert this in a black-box manner to a group privacy guarantee. For example, any $f$-DP mechanism satisfies $f'$-DP with respect to add-or-remove-up-to-$k$ where $f'(x) = 1 - (1 - f)^{\circ k}(x)$, $\circ k$ denotes self-composition of a function $k$ times. However, black-box approaches to group privacy are generally loose. Especially for mechanisms like the exponential mechanism, a black-box privacy guarantee might be pointwise tight everywhere, as in at each FPR-FNR pair there is an instantiation of the mechanism and two inputs $S, S'$ that achieve that FPR-FNR pair, but not ``globally'' tight in the sense that a single mechanism and pair of inputs achieves all the FPR-FNR pairs simultaneously. A black-box group privacy analysis might try to combine information at different FPR-FNR points that are not ``globally'' tight, thus giving an upper bound that no instantiation of the mechanism and inputs can realize, whereas a ``white-box'' group privacy analysis looking directly at $M(S), M(S')$ for $S \simeq S'$ would avoid this issue entirely.

\subsection{Group Privacy of Homogenous Compositions of Exponential Mechanisms}
\label{sec:tech_hierarchical_group}
In this section, we give a group privacy guarantee of the homogenous composition of Exponential mechanisms, which is the case of the DP quantiles~\cite{kaplan2022differentially} in \cref{alg:quantiles}. We give a privacy guarantee under ``add-or-remove-up-to-$k$'' adjacency $\simeq_{AR \leq k}$, i.e. $S \simeq_{AR \leq k} S'$ if $S \supseteq S'$ and $S$ contains at most $k$ elements not in $S$ (or vice-versa).
\begin{theorem}\label{thm:brgroupprivacy}
    Any $\nu$-BR mechanism with respect to the add-or-remove-one adjacency is $k \nu$-BR with respect to $\simeq_{AR \leq k}$.
\end{theorem}
\begin{proof}
    Let $S \simeq_{AR \leq k} S'$, assuming without loss of generality $|S \setminus S'| = k$, and let $S_0 = S, S_1, S_2, \ldots, S_k = S'$ be a sequence of datasets where $|S_i - S_{i+1}| = 1$. As the Exponential mechanism is $\nu$-BR, for any output set $E$ and any $i$ we have $\log\left(\frac{\Pr[M(S_i) \in E]}{\Pr[M(S_{i+1}) \in E]}\right) \in [-t_i, -t_i + \nu]$ for some $t_0, t_1, \ldots, t_{k-1}$. In turn, we have:

    \begin{align*}
        \log\left(\frac{\Pr[M(S) \in E]}{\Pr[M(S') \in E]}\right) &= \log\left(\frac{\Pr[M(S_0) \in E]}{\Pr[M(S_k) \in E]}\right) \\
        &= \sum_{i = 0}^{k-1}\log\left(\frac{\Pr[M(S_i) \in E]}{\Pr[M(S_{i+1}) \in E]}\right), \\
    \end{align*}
    which is in $\left[-\sum_{i=0}^{k-1} t_i, -\sum_{i=0}^{k-1} t_i + k\nu\right].$
\end{proof}
\begin{theorem}\label{thm:brpartition}
    Let $g : \mathbb{D} \rightarrow [m]$ be any function from the data universe to $[m]$. Let $M$ be the following mechanism: Given the input $S$, partition $S$ into $\{S_1, S_2, \ldots S_m\}$ where $S_i = \{x \in S: g(x) = i\}$, then return $(M_1(S_1), M_2(S_2), \ldots, M_m(S_m))$ where each $M_i$ is an (independent) $\nu$-BR mechanism (with respect to add-or-remove-one). Then $M$ is $k \nu$-BR with respect to $\simeq_{AR \leq k}$.
\end{theorem}
\begin{proof}
    Let $S \simeq_{AR \leq k} S'$, assuming without loss of generality $|S \setminus S'| = k$. Let $k_i$ be the number of elements in $S \setminus S'$ in $S_i$. By \cref{thm:brgroupprivacy}, for some $t_1, t_2, \ldots t_m$ and any $i$, output $o_i$ for $M_i$ we have \[\log\left(\frac{M_i(S)(o_i)}{M_i(S')(o_i)}\right) \in [-t_i, -t_i + k_i\nu].\]
    Since the mechanisms are run independently and the elements in $S'$ are assigned to the same subset of the partition whether we receive input $S$ or $S'$, we have for any output $o = (o_1, o_2, \ldots o_m)$ for $M$:
    \[\log\left(\frac{M(S)(o)}{M(S')(o)}\right) = \sum_{i=1}^m \log\left(\frac{M_i(S)(o_i)}{M_i(S')(o_i)}\right)\]
    \[\in \left[-\sum_{i=1}^m t_i, -\sum_{i=1}^m t_i + \nu \sum_{i=1}^m k_i\right] = \left[-\sum_{i=1}^m t_i, -\sum_{i=1}^m t_i + k \nu\right].\]
\end{proof}
\begin{theorem}
\cref{alg:quantiles} satisfies any privacy guarantee satisfied by the composition of a  $k\varepsilon_1$-BR, $k\varepsilon_2$-BR, \ldots, and $k\varepsilon_L$-BR mechanisms.    
\end{theorem}
\begin{proof}
     We can assume the mechanism is run in a ``top-down'' fashion, i.e. we run all mechanisms at $\ell = i$ before running any mechanism at level $\ell = i - 1$. Viewing the concatenation of all mechanisms being run at $\ell = i$ as a single mechanism, \cref{alg:quantiles} is then just an adaptive composition of the per-level mechanisms. Each per-level mechanism is exactly a partition of the dataset and a run of the exponential mechanism on each subset of the partition, exactly as described/analyzed by \cref{thm:brpartition}, giving the theorem.
\end{proof}

\section{Algorithms and Pipeline Details}
\label{app:extra-details}

In this section, we provide additional algorithmic and pipeline details for the methods introduced in \cref{sec:applications}, including AIM-GDP, DP preprocessing, and MixCal.

For reference, \cref{tab:mechanism_taxonomy} summarizes the DP mechanisms relevant to data preprocessing and tabular synthesis considered in our pipelines.

\begin{table*}[t!]
\centering
\scriptsize
\setlength{\tabcolsep}{6pt}
\renewcommand{\arraystretch}{1.15}
\begin{tabular}{lllll}
\toprule
\textbf{Mechanism} & \textbf{Param} & \textbf{Base Events} & \textbf{Fully Adaptive} & \textbf{Accounting} \\
\midrule
\multicolumn{5}{l}{\textit{\textbf{Preprocessing \& Domain Discovery}}} \\
DP Quantiles \citep{gillenwater} & $\varepsilon$ & Exponential & No & ADP \\
DP Quantiles \citep{kaplan2022differentially} & $\varepsilon_1, \dots, \varepsilon_l$ & Exponential & No & zCDP \\
Noisy Histogram & $\sigma$ & Gaussian & No & GDP \\
Threshold Gaussian \citep{gopi20setunion} & $\sigma, \delta$ & \begin{tabular}[t]{@{}l@{}}Gaussian \\ Failure\end{tabular} & No & GDP \\
Optimal Partition \citep{desfontaines2022differentially} & $\varepsilon, \delta$ & \begin{tabular}[t]{@{}l@{}}Laplace \\ Failure\end{tabular} & No & ADP \\
RDP Partition \citep{harrison2026optimal} & $\sigma, \delta$ & \begin{tabular}[t]{@{}l@{}}Gaussian \\ Failure\end{tabular} & No & RDP \\
\midrule
\multicolumn{5}{l}{\textit{\textbf{Base Mechanism}}} \\
AIM \citep{mckenna2022aim} & $\rho$ & \begin{tabular}[t]{@{}l@{}}Exponential \\ Gaussian\end{tabular} & Yes & zCDP \\
AIM-GDP (\cref{alg:aimgdp}) & $\mu$ & \begin{tabular}[t]{@{}l@{}}Exponential \\ Gaussian\end{tabular} & Yes & GDP \\
GEM \citep{liu2021iterative} & $\rho$ & \begin{tabular}[t]{@{}l@{}}Exponential \\ Gaussian\end{tabular} & No & zCDP \\
RAP \citep{aydore2021rap} & $\rho$ & \begin{tabular}[t]{@{}l@{}}Exponential \\ Gaussian\end{tabular} & No & zCDP \\
PrivSyn \citep{zhang2021privsyn} & $\rho$ & Gaussian & No & zCDP \\
MST \citep{mckenna2021relaxed} & $\rho$ & \begin{tabular}[t]{@{}l@{}}Exponential \\ Gaussian\end{tabular} & No & zCDP \\
\midrule
\multicolumn{5}{l}{\textit{\textbf{Language Model Finetuning}}} \\
Poisson DP-SGD~\cite{abadi2016deep} & $\sigma$ & Sub. Gaussian & No & RDP, PLD \\
Random Allocation DP-SGD~\cite{feldman2025privacy} & $\sigma$ & Sub. Gaussian & No & PLD \\
DP-BandMF~\cite{mckenna2024scaling} & $\sigma$ & Corr. Gaussian & No & Monte Carlo \\
\bottomrule
\end{tabular}
\caption{Classification of differentially private synthetic data mechanisms across domain discovery, tabular synthesis, and language model finetuning pillars.}
\label{tab:mechanism_taxonomy}
\end{table*}

\subsection{AIM-GDP}

\colorlet{zcdpcol}{orange!25}
\colorlet{gdpcol}{teal!20}
\newcommand{\zc}[1]{\colorbox{zcdpcol}{#1}}  %
\newcommand{\gd}[1]{\colorbox{gdpcol}{#1}}   %

\begin{figure*}[t!]
\begin{minipage}[t]{0.49\textwidth}
\centering
\resizebox{0.71\linewidth}{!}{%
\begin{minipage}[t]{\textwidth}
\begin{algorithm}[H]
    \caption{\label{alg:aim} \aim}
\begin{algorithmic}[1]
    \State \textbf{Input:} Dataset $D$, workload $W$, \zc{privacy parameter $\rho$ (zCDP)}
    \State \textbf{Output:} Synthetic dataset $\hat{D}$
    \State \textbf{Hyper-Parameters:} \textsf{MAX-SIZE}=80MB, $T=16d$, $\alpha=0.9$
    \State \zc{$\sigma_0 = \sqrt{T / (2 \, \alpha \, \rho)}$}
    \State $\rho_\mathrm{used} = 0$
    \State $t = 0$
    \State Initialize $\hat{p}_t$ using \cref{alg:init} \label{line:init}
    \State $w_r = \sum_{s \in W} c_s \mid r \cap s \mid$
    \State $\sigma_{t+1} \leftarrow \sigma_0$ \quad \zc{$\nu_{t+1} \leftarrow \sqrt{8 (1 - \alpha) \rho / T}$}
    \While{$\rho_\mathrm{used} < \rho$} \label{line:loop}
        \State $t = t+1$
        \State $\rho_\mathrm{used} \leftarrow \rho_\mathrm{used} + {}$\zc{$\frac{1}{8} \nu_t^2 + \frac{1}{2\sigma_t^2}$} \label{line:budget}
        \State $C_t = \set{ r_t \in W_+ \mid \textsf{JT-SIZE}(r_1, \dots, r_t) \leq \frac{\rho_\mathrm{used}}{\rho} \cdot \textsf{MAX-SIZE} }$ \label{line:candidates}
        \State \textbf{select} $r_t \in C_t$ using the exponential mechanism with:
$$ q_r(D) = w_r \Big( \norm{ M_r(D) - M_r(\hat{p}_{t-1}) }_1 - \sqrt{2/\pi} \cdot \sigma_t \cdot n_r \Big) $$ \label{line:quality}
        \State \textbf{measure} marginal on $r_t$:
$$ \tilde{y}_t = M_{r_t}(D) + \mathcal{N}(0, \sigma_t^2 \mathbb{I}) $$
        \State \textbf{estimate} data distribution using \pgm:
$$ \hat{p}_t = \argmin_{p \in S} \sum_{i=1}^t \frac{1}{\sigma_i} \norm{ M_{r_i}(p) - \tilde{y}_i }_2^2 $$
        \State anneal $\nu_{t+1}$ and $\sigma_{t+1}$ using \cref{alg:anneal} \label{line:anneal}
    \EndWhile
    \State \textbf{generate} synthetic data $\hat{D}$ from $\hat{p}_t$ using \pgm
    \State \textbf{return} $\hat{D}$
\end{algorithmic}
\end{algorithm}

\begin{algorithm}[H]
    \caption{\label{alg:init} Initialize $p_t$ (subroutine of \cref{alg:aim})}
\begin{algorithmic}[1]
    \For{$r \in \set{r \in W_+ \mid |r| = 1}$}
        \State $t = t+1$ \quad $\sigma_t \leftarrow \sigma_0$ \quad $r_t \leftarrow r$
        \State $\tilde{y}_t = M_r(D) + \mathcal{N}(0, \sigma_t^2 \mathbb{I})$
        \State $\rho_\mathrm{used} \leftarrow \rho_\mathrm{used} + {}$\zc{$\frac{1}{2 \sigma_t^2}$}
    \EndFor
    \State $\hat{p}_t = \argmin_{p \in S} \sum_{i=1}^t \frac{1}{\sigma_i} \norm{ M_{r_i}(p) - \tilde{y}_i }_2^2$
\end{algorithmic}
\end{algorithm}

\begin{algorithm}[H]
    \caption{\label{alg:anneal} Budget annealing (subroutine of \cref{alg:aim})}
\begin{algorithmic}[1]
    \If{$\norm{M_{r_t}(\hat{p}_t) - M_{r_t}(\hat{p}_{t-1})}_1 \leq \sqrt{2/\pi} \cdot \sigma_t \cdot n_{r_t}$}
        \State $\nu_{t+1} \leftarrow 2 \cdot \nu_t$
        \State $\sigma_{t+1} \leftarrow \sigma_t / 2$
    \Else
        \State $\nu_{t+1} \leftarrow \nu_t$
        \State $\sigma_{t+1} \leftarrow \sigma_t$
    \EndIf
    \If{$(\rho - \rho_\mathrm{used}) \leq 2 \big($\zc{$\frac{1}{2\sigma_{t+1}^2} + \frac{1}{8} \nu_{t+1}^2$}$\big)$} \label{line:endcond}
        \State \zc{$\nu_{t+1} = \sqrt{8 \cdot (1-\alpha) \cdot (\rho - \rho_\mathrm{used})}$}
        \State \zc{$\sigma_{t+1} = \sqrt{1 / (2 \cdot \alpha \cdot (\rho - \rho_\mathrm{used}))}$}
    \EndIf
\end{algorithmic}
\end{algorithm}
\end{minipage}}%
\end{minipage}%
\hfill
\begin{minipage}[t]{0.49\textwidth}
\centering
\resizebox{0.71\linewidth}{!}{%
\begin{minipage}[t]{\textwidth}
\begin{algorithm}[H]
    \caption{\label{alg:aimgdp} \aimgdp}
\begin{algorithmic}[1]
    \State \textbf{Input:} Dataset $D$, workload $W$, \gd{privacy parameter $\smash{\mu^2}$ (GDP)}
    \State \textbf{Output:} Synthetic dataset $\hat{D}$
    \State \textbf{Hyper-Parameters:} \textsf{MAX-SIZE}=80MB, $T=16d$, $\alpha=0.9$
    \State \gd{$\sigma_0 = \sqrt{T / (\alpha \, \smash{\mu^2})}$}
    \State $\smash{\mu^2}_\mathrm{used} = 0$
    \State $t = 0$
    \State Initialize $\hat{p}_t$ using \cref{alg:init-gdp} \label{gline:init}
    \State $w_r = \sum_{s \in W} c_s \mid r \cap s \mid$
    \State $\sigma_{t+1} \leftarrow \sigma_0$ \quad \gd{$\nu_{t+1} \leftarrow b^{-1}\big( (1 - \alpha) \smash{\mu^2} / T \big)$}
    \While{$\smash{\mu^2}_\mathrm{used} < \smash{\mu^2}$} \label{gline:loop}
        \State $t = t+1$
        \State $\smash{\mu^2}_\mathrm{used} \leftarrow \smash{\mu^2}_\mathrm{used} + {}$\gd{$b(\nu_t) + \frac{1}{\sigma_t^2}$} \label{gline:budget}
        \State $C_t = \set{ r_t \in W_+ \mid \textsf{JT-SIZE}(r_1, \dots, r_t) \leq \frac{\smash{\mu^2}_\mathrm{used}}{\smash{\mu^2}} \cdot \textsf{MAX-SIZE} }$ \label{gline:candidates}
        \State \textbf{select} $r_t \in C_t$ using the exponential mechanism with:
$$ q_r(D) = w_r \Big( \norm{ M_r(D) - M_r(\hat{p}_{t-1}) }_1 - \sqrt{2/\pi} \cdot \sigma_t \cdot n_r \Big) $$ \label{gline:quality}
        \State \textbf{measure} marginal on $r_t$:
$$ \tilde{y}_t = M_{r_t}(D) + \mathcal{N}(0, \sigma_t^2 \mathbb{I}) $$
        \State \textbf{estimate} data distribution using \pgm:
$$ \hat{p}_t = \argmin_{p \in S} \sum_{i=1}^t \frac{1}{\sigma_i} \norm{ M_{r_i}(p) - \tilde{y}_i }_2^2 $$
        \State anneal $\nu_{t+1}$ and $\sigma_{t+1}$ using \cref{alg:anneal-gdp} \label{gline:anneal}
    \EndWhile
    \State \textbf{generate} synthetic data $\hat{D}$ from $\hat{p}_t$ using \pgm
    \State \textbf{return} $\hat{D}$
\end{algorithmic}
\end{algorithm}

\begin{algorithm}[H]
    \caption{\label{alg:init-gdp} Initialize $p_t$ (subroutine of \cref{alg:aimgdp})}
\begin{algorithmic}[1]
    \For{$r \in \set{r \in W_+ \mid |r| = 1}$}
        \State $t = t+1$ \quad $\sigma_t \leftarrow \sigma_0$ \quad $r_t \leftarrow r$
        \State $\tilde{y}_t = M_r(D) + \mathcal{N}(0, \sigma_t^2 \mathbb{I})$
        \State $\smash{\mu^2}_\mathrm{used} \leftarrow \smash{\mu^2}_\mathrm{used} + {}$\gd{$\frac{1}{\sigma_t^2}$}
    \EndFor
    \State $\hat{p}_t = \argmin_{p \in S} \sum_{i=1}^t \frac{1}{\sigma_i} \norm{ M_{r_i}(p) - \tilde{y}_i }_2^2$
\end{algorithmic}
\end{algorithm}

\begin{algorithm}[H]
    \caption{\label{alg:anneal-gdp} Budget annealing (subroutine of \cref{alg:aimgdp})}
\begin{algorithmic}[1]
    \If{$\norm{M_{r_t}(\hat{p}_t) - M_{r_t}(\hat{p}_{t-1})}_1 \leq \sqrt{2/\pi} \cdot \sigma_t \cdot n_{r_t}$}
        \State $\nu_{t+1} \leftarrow 2 \cdot \nu_t$
        \State $\sigma_{t+1} \leftarrow \sigma_t / 2$
    \Else
        \State $\nu_{t+1} \leftarrow \nu_t$
        \State $\sigma_{t+1} \leftarrow \sigma_t$
    \EndIf
    \If{$(\smash{\mu^2} - \smash{\mu^2}_\mathrm{used}) \leq 2 \big($\gd{$\frac{1}{\sigma_{t+1}^2} + b(\nu_{t+1})$}$\big)$} \label{gline:endcond}
        \State \gd{$\nu_{t+1} = b^{-1}\big( (1-\alpha) \cdot (\smash{\mu^2} - \smash{\mu^2}_\mathrm{used}) \big)$}
        \State \gd{$\sigma_{t+1} = \sqrt{1 / (\alpha \cdot (\smash{\mu^2} - \smash{\mu^2}_\mathrm{used}))}$}
    \EndIf
\end{algorithmic}
\end{algorithm}
\end{minipage}}%
\end{minipage}
\caption{\label{fig:aim-side-by-side}
Left: the original \aim with zCDP accounting
(\smash{\colorbox{zcdpcol}{orange}}). Right: \aimgdp, the same algorithm with the
budget tracked in GDP units $\smash{\mu^2}$, which compose additively
(\smash{\colorbox{gdpcol}{teal}}). The Gaussian mechanism with noise scale $\sigma$
is exactly $\nicefrac{1}{\sigma}$-GDP, and the exponential mechanism with parameter
$\nu$ is a $\nu$-BR mechanism with exact squared GDP budget 
$b(\nu) = \big( 2 \Phi^{-1}\big( \tfrac{1}{e^{\nu/2} + 1} \big) \big)^2$ given by \cref{stmt:br-to-gdp}, with
$b^{-1}(\smash{\mu^2}) = 2 \ln \frac{\Phi(\nicefrac{\mu}{2})}{\Phi(-\nicefrac{\mu}{2})}$ denoting its
inverse used for calibrating the parameter to a given budget.}
\end{figure*}

\crefrange{alg:aim}{alg:anneal-gdp} present the original AIM and our AIM-GDP variant.
We introduce AIM-GDP as the base synthetic data mechanism in \cref{sec:applications} and evaluate it in \cref{sec:exp}, both in isolation and as a component of the end-to-end and multimodal pipelines.

\subsection{DP Preprocessing}

\subsubsection{DP Quantiles}
\cref{alg:quantiles} shows the recursive-median DP quantile mechanism from \citet{kaplan2022differentially} used to discretize numerical attributes in our end-to-end pipeline. We describe it in \cref{sec:applications}, and evaluate the pipeline in \cref{sec:exp}.

\subsubsection{DP Open-Set Partition Selection}
\cref{alg:open_set} presents the Threshold Gaussian~\cite{gopi20setunion} partition-selection mechanism used for private domain discovery of open-set categorical attributes. We describe it in \cref{sec:applications}, and evaluate the pipeline in \cref{sec:exp}.

\begin{algorithm}[t!]
\caption{DP Quantile Discretization via Recursive Medians}
\label{alg:quantiles}
\begin{algorithmic}[1]
\Require 1D histogram $\mathbf{c} = (c_0, \dots, c_{K-1}) \in \mathbb{N}^K$, tree depth $L$, per-level privacy budgets $\{\epsilon_\ell\}_{\ell=1}^L$.
\Ensure Set of discretized bin cutpoints $\mathcal{E} \subset \{0, \dots, K-1\}$.
\Statex
\Function{RecursiveBisection}{$\text{lo}, \text{hi}, \ell$}
    \If{$\ell = 0$}
        \State \Return $\emptyset$
    \EndIf
    \State $r \gets \frac{1}{2} \sum_{j=\text{lo}}^{\text{hi}-1} c_j$ \Comment{Target median rank within current sub-interval}
    \For{each candidate split $v \in \{\text{lo}, \dots, \text{hi}-1\}$}
        \State $L_v \gets \sum_{j=\text{lo}}^{v-1} c_j, \quad R_v \gets \sum_{j=\text{lo}}^v c_j$ \Comment{Cumulative rank bounds of cell $v$}
        \State $u(v) \gets -\max\left(0,\, L_v - r,\, r - R_v\right)$ \Comment{Distance to target median rank}
    \EndFor
    \State Sample $v^* \in \{\text{lo}, \dots, \text{hi}-1\}$ w.p.  $\propto \exp\left(\epsilon_\ell \cdot u(v)\right)$ \Comment{Exp.\ Mech.\ with $\Delta u = 1/2$}
    \State $\mathcal{E}_{\text{left}} \gets \Call{RecursiveBisection}{\text{lo}, v^*, \ell - 1}$
    \State $\mathcal{E}_{\text{right}} \gets \Call{RecursiveBisection}{v^*, \text{hi}, \ell - 1}$
    \State \Return $\mathcal{E}_{\text{left}} \cup \{v^*\} \cup \mathcal{E}_{\text{right}}$
\EndFunction
\Statex
\State \Return $ \Call{RecursiveBisection}{0, K, L}$
\end{algorithmic}
\end{algorithm}

\begin{algorithm}[t!]
\caption{DP Open-Set Partition Selection via Threshold Gaussian}
\label{alg:open_set}
\begin{algorithmic}[1]
\Require Partition counts $\mathbf{y} \in \mathbb{N}^{|\mathcal{V}|}$ over active keys $\mathcal{V} = \{k : y_k \ge \tau_{\min}\}$, noise scale $\sigma$, target failure probability $\delta$, minimum threshold $\tau_{\min}$.
\Ensure Discovered vocabulary $\widehat{\mathcal{V}}$, noisy partition counts $\tilde{\mathbf{y}}$, and closed-domain mapping $\mathcal{M}$.
\Statex
\State $T \gets \tau_{\min} + \sigma \Phi^{-1}(1 - \delta)$ \Comment{Threshold cutoff bounding false positives to $\le \delta$ per empty bin}
\State $\widehat{\mathcal{V}} \gets \emptyset$, \quad $\tilde{\mathbf{y}} \gets \emptyset$
\For{each active partition key $k \in \mathcal{V}$}
    \State Sample $\xi_k \sim \mathcal{N}(0, \sigma^2)$
    \State $\tilde{y}_k \gets y_k + \xi_k$ \Comment{Perturb count with Gaussian noise}
    \If{$\tilde{y}_k \ge T$}
        \State $\widehat{\mathcal{V}} \gets \widehat{\mathcal{V}} \cup \{k\}$ \Comment{Retain key in discovered vocabulary}
        \State $\tilde{\mathbf{y}}[k] \gets \tilde{y}_k$
    \EndIf
\EndFor
\State Closed vocabulary domain $\mathcal{D} \gets \widehat{\mathcal{V}} \cup \{\texttt{[OTHER]}\}$
\State Construct mapping $\mathcal{M}(k) = k$ if $k \in \widehat{\mathcal{V}}$ else $\texttt{[OTHER]}$ \Comment{Collapse unselected tail categories}
\State \Return Discovered vocabulary $\widehat{\mathcal{V}}$, noisy counts $\tilde{\mathbf{y}}$, mapping $\mathcal{M}$
\end{algorithmic}
\end{algorithm}

\subsubsection{Measurement Recycling for End-to-End Synthetic Data}
Beyond tight accounting, we can use an additional enhancement that allows us to skip the dedicated one-way marginal measurement step in \aimgdp by  reusing information measured during the preprocessing step of mixed tabular domains (open set discovery and quantile-based discretization). A na\"ive implementation would treat preprocessing as disjoint from the application of the base mechanism by spending a fraction of the privacy budget to establish discrete domains, and a separate fraction of the privacy budget measuring one-way marginals. This is wasteful, as it pays twice for single-attribute information and diminishes the remaining budget available for capturing multivariate interactions. Rather than measuring one-way marginals again, we eliminate this step entirely by recycling the private information obtained during preprocessing. 

Under the hood, synthetic data mechanisms like \aim rely on estimation algorithms like \pgm~\citep{mckenna2019graphical}. These attempt to estimate the discrete data distribution $\hat{p}$ given noisy linear measurements of its marginals of the form $y = Q p + \xi$, where $p$ denotes the true data vector, $Q$ is a linear query matrix, and $\xi$ is zero-centered noise with standard deviation $\sigma$. Whereas closed-set categorical attributes directly yield complete noisy histograms and thus total dataset size estimates, open-set and numerical attributes provide only partial information that we treat as follows:
\begin{itemize}[leftmargin=*]
    \item For open-set categories with discovered frequent items $\widehat{\mathcal{V}} = \{k_1, \dots, k_m\}$ and an unmeasured tail bucket at index $0$, partition selection provides noisy counts $\tilde{y} \in \mathbb{R}^m$. We incorporate these as a partial linear measurement with query matrix $Q_{\text{open}} = \begin{bmatrix} 0_{m \times 1} & I_m \end{bmatrix}$, which depends only on the observed categories while leaving the unmeasured tail unconstrained.
    \item For continuous attributes handled via quantile discretization, the distribution over bins is approximately uniform by definition of quantiles. Given an estimate of the total dataset size $\tilde{N}$, we construct an identity query $Q_{\text{num}} = I_K$ with target $y_{\text{num}} = \nicefrac{\tilde{N}}{K} \cdot \indic{}_K$ and (heuristic) uncertainty $\sigma_{\text{num}} \propto \nicefrac{1}{\sqrt{\rho}}$ calibrated to the $\rho$-zCDP quantile budget.
\end{itemize}
Accounting for precisely what was measured during preprocessing allows downstream models to initialize one-way marginals for free, preserving the entire synthesis budget for modeling higher-order dependencies.

\begin{figure}[t!]
\centering
    \includegraphics[width=0.66\textwidth]{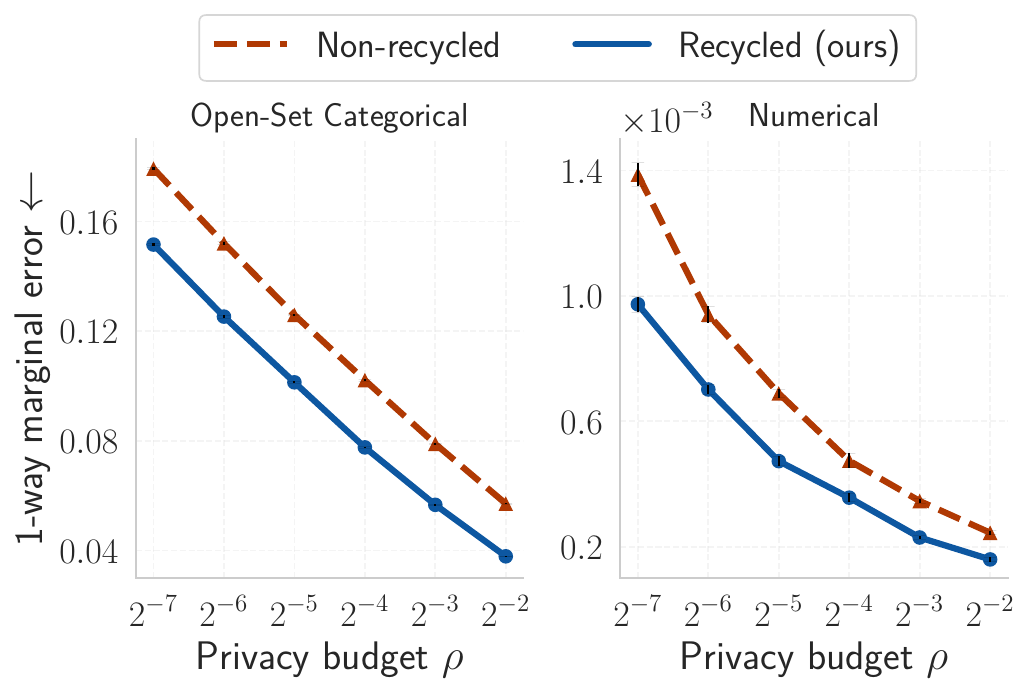}
    \caption{One-way marginal measurement recycling after preprocessing.}
    \label{fig:recycling}
\end{figure}

When applied to 1D data, the improvement from measurement recycling was easily observable and matches expectations, as shown in \cref{fig:recycling}.
When we applied this to an end-to-end synthetic data mechanism like \aim, however, we observed no meaningful difference in performance between the recycling approach and the naive baseline that splits the budget between preprocessing and one-way marginal measurement. We therefore consider this a \emph{negative} result. Future work, however, could investigate whether there exist settings where marginal recycling provides benefits.

\subsection{\mixcal}

\begin{algorithm}[t!]
\caption{\mixcal}
\label{alg:mixcal}
\begin{algorithmic}[1]
\Require Target guarantee $(\varepsilon^{\star}, \delta^{\star})$; mechanism family $M_i(\cdot)$ each
  parameterized by additive budget $b_i$;
  weights $w_1, \dots, w_k$ with $\sum_i w_i = 1$
\Statex
\Function{PLDAccounting$_{\delta}$}{$b$} \Comment{privacy cost of the composition at budget $b$}
  \For{$i = 1, \dots, k$}
    \State $b_i \gets w_i \, b$ \Comment{allocate the budget additively}
    \State Calibrate $M_i(\cdot)$ to satisfy $b_i$ privacy.
  \EndFor
  \State \Return $\varepsilon$ of $M_1 \circ \cdots \circ M_k$ at $\delta^{*}$
\EndFunction
\Statex
\State Binary search for the largest $b$ with $\textsc{PLDAccounting}_{\delta^\star}(b) \le \varepsilon^{\star}$.
\State \Return a calibrated mechanism $M = M_1 \circ \cdots \circ M_k$.
\end{algorithmic}
\end{algorithm}

\cref{alg:mixcal} formalizes \mixcal, which we introduce in \cref{sec:applications} to calibrate the component mechanisms of heterogeneous compositions while retaining tight PLD accounting.

In \mixcal, we introduce a budget parameter $b$ that should represent an additive privacy notation like zCDP or GDP. Suppose that a budget split for a given component is $b_i$. Then, in \mixcal, we calibrate the component to satisfy this privacy guarantee $b_i$. 

A practically useful property of \mixcal is that the correctness of the final guarantee only depends on the correctness of PLD accounting. Therefore, intermediate calibration steps to budget shares need only be approximate. For instance, when using GDP as the budget notion, we can calibrate a Gaussian mechanism exactly to its budget share $b_i = \mu_i$ by setting $\sigma = \nicefrac{1}{\mu_i}$, but when calibrating Threshold Gaussian with a relatively small probability of failure, e.g., $\tau < 10^{-5}$, using the results of \cref{app:extra-methods}, we can ignore the failure term and treat the mechanism as Gaussian for the purposes of budget allocation in \mixcal.

For instance, using a zCDP interpretation of $b_i$, we could use:
\begin{itemize}
\item Gaussian Mechanism: $\sigma = \sqrt{\frac{1}{2b_i}}$ (exact)
\item Exponential Mechanism $\nu = \sqrt{8 b_i}$ (pessimistic)
\item Threshold Gaussian: $\sigma = \sqrt{\frac{1}{2b_i}}$ (approximate)
\item DP-SGD noise multiplier: $\sigma = \sqrt{\frac{T}{2b_i}}$, where $T$ is the number of iterations (approximate; assuming no amplification).
\end{itemize}

\section{Additional Experiment Details}

\label{app:data}
\label{app:riskcal}
\label{app:audit}
\label{app:extra-exp}

\subsection{Details of the Accounting Comparisons}

In this section, we describe the setups behind \cref{fig:gdp-demo}. In both panels, we report $\varepsilon_{\text{zCDP}} - \varepsilon_{\text{ours}}$ at $\delta = 10^{-5}$, so positive values mean that our analysis is tighter. For the zCDP analysis of the Exponential mechanism, we use the tightest available expression due to \citet{harrison2025exact}:
$$\rho_{\text{BR}}(\nu) = \frac{\nu}{e^\nu + 1} + \log\left( \frac{e^\nu - 1}{\nu}\right) - 1,$$
and convert zCDP to $(\varepsilon, \delta)$-DP with the standard numerical conversion from \citet{canonne2020discrete}. The GDP parameter $\mu_{\text{BR}}(\nu)$ of a $\nu$-BR mechanism is given by \cref{stmt:br-to-gdp}.

\paragraph{Left: GDP vs. zCDP in fully adaptive compositions.}
We consider $T \in \{1, 5, 10, 20, 50, 100\}$ fully adaptive compositions of a Gaussian and an Exponential mechanism, which represents \aim-like workloads. We fix the zCDP cost of each pair to $\rho = 0.01$ and split it according to the proportion $p \in [0.1, 0.99]$ allocated to the Gaussian mechanism. 

To split the budget between the two mechanisms, we follow the \mixcal approach outlined in \cref{app:extra-details}. Using the standard analyses~\cite{cesar2021bounding,durfee2019practical} for simplicity, we set the Gaussian parameter to $\mu = \sqrt{2 \rho}$, and the Exponential mechanism parameter to $\nu = \smash{\sqrt{8 \rho (1 - p)}}$. 

For the final accounting, we use tight expressions. The GDP analysis composes the pairs as:
$$\mu_{\text{total}} = \sqrt{T \cdot (\mu^2 + \mu_{\text{BR}}(\nu)^2)},$$
and the zCDP analysis uses:
$$\rho_{\text{total}} = T \cdot \left(\frac{\mu^2}{2} + \rho_{\text{BR}}(\nu)\right).$$
We find that the GDP analysis is tighter when the majority of the budget of each pair is consumed by the Gaussian mechanism (e.g., in \aim, the default is $p = 0.9$). This directly translates into improvements in the performance of \aimgdp, as we show in \cref{sec:exp-aimgdp}.

\paragraph{Right: PLD vs. zCDP for preprocessing workloads.}
We consider a preprocessing-like workload with $T \in \{1, \dots, 49\}$ features, composed non-fully-adaptively. Each feature is processed by one Exponential mechanism that satisfies $\nu$-BR, and one Threshold Gaussian mechanism. For simplicity, we set $\nu = \mu \in \{0.05, 0.1, 0.15, 0.2\}$ for all features, where $\mu$ is the GDP parameter of the Gaussian part, so a larger $\mu$ means less noise. We model the Threshold Gaussian mechanism as a composition of the Gaussian mechanism and a $(0, \tau)$-DP mechanism with $\tau = 10^{-8}$, which accounts for the failure probability of the threshold (see \cref{app:extra-methods}).

For our analysis, we convolve the PLDs of all mechanisms: $T$ copies of the PLD from \cref{stmt:br-to-pld} with parameter $\nu$, and $T$ copies of the Gaussian PLD composed with the $(0, \tau)$-DP event as per the analysis in \cref{app:extra-methods}.
For the zCDP baseline, the total budget is $\rho = T \cdot (\nicefrac{\mu^2}{2} + \rho_{\text{BR}}(\nu))$, with a corresponding accounting for the failure probability.

\subsection{Datasets}

For our discrete and end-to-end synthetic data experiments, we use four standard datasets from the machine learning literature that have been extensively used in prior work on DP synthetic data~\cite[see, e.g.,][]{mckenna2022aim, tao2022benchmarking, annamalai2024what, ganev2025importance}: Adult, Fire, Page blocks, and Gas.
A summary is provided in \cref{tab:datasets}.
We use two versions of each dataset: 1) a fully discretized version, using uniform discretization into 16 bins, for our discrete experiments, and 2) the original, unprocessed version for our end-to-end experiments.

For our multimodal experiments, we use the bioRxiv dataset, a corpus of scientific abstracts ($n = 34k$, average number of
tokens per sample around 300). We use the filtered dataset from \cite{hu2026actg}, that is, we keep only the abstracts that appear after the knowledge cutoff date of Gemma-3 family models (Aug 2024). We perform train/validation/test split and obtain a train set of size $n = 28,846$. 

\begin{table}[t!]
\footnotesize
\centering
\caption{Datasets used in our discrete/end-to-end synthetic data experiments. $n$ is the number of training records, $d$ the number of columns, and $d_{\mathrm{num}}$ and $d_{\mathrm{cat}}$ the numbers of numerical and categorical columns, respectively. The final column reports the categorical column with the largest domain and its domain size.}
\label{tab:datasets}
\begin{tabular}{lrrrrl}
\toprule
\textbf{Dataset}
& $n$
& $d$
& $d_{\mathrm{num}}$
& $d_{\mathrm{cat}}$
& \textbf{Largest categorical domain} \\
\midrule
Adult  & 32,561  & 11 &  0 & 11 & native\_country (42) \\
Fire             & 244,106 & 10 &  0 & 10 & Station Area (46) \\
Page blocks      & 4,378   & 11 & 10 &  1 & block (2) \\
Gas              & 22,191  & 11 & 10 &  1 & TEY (2) \\
\bottomrule
\end{tabular}
\end{table}

\subsection{Discrete Synthetic Data with \aimgdp}
\label{app:emp_b}

Detailed results for all four (discrete) datasets are shown in \cref{fig:aimgdp_full}, with additional diagnostic experiments on the Adult dataset presented in \cref{fig:aimgdp_diagnostics,fig:aimgdp_network}.
We discuss these experiments in \cref{sec:exp-aimgdp}.

\begin{figure}[t!]
    \centering
    \includegraphics[width=0.45\linewidth]{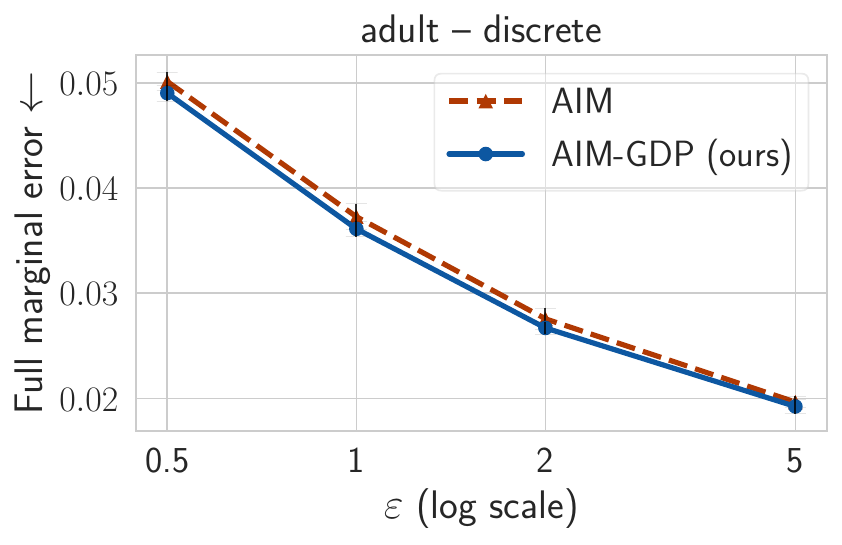}
    \vspace{0.5em}
    \includegraphics[width=0.45\linewidth]{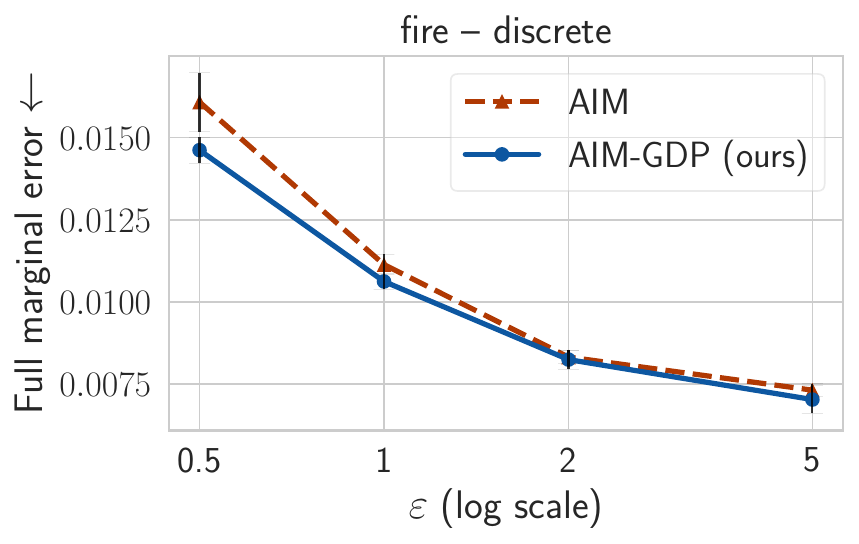}
    \hfill
    \includegraphics[width=0.45\linewidth]{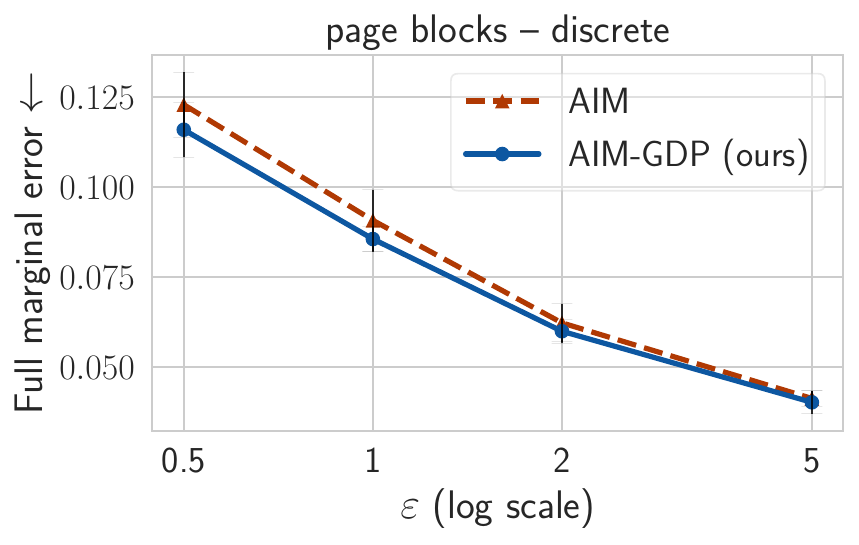}
    \vspace{0.5em}
    \includegraphics[width=0.45\linewidth]{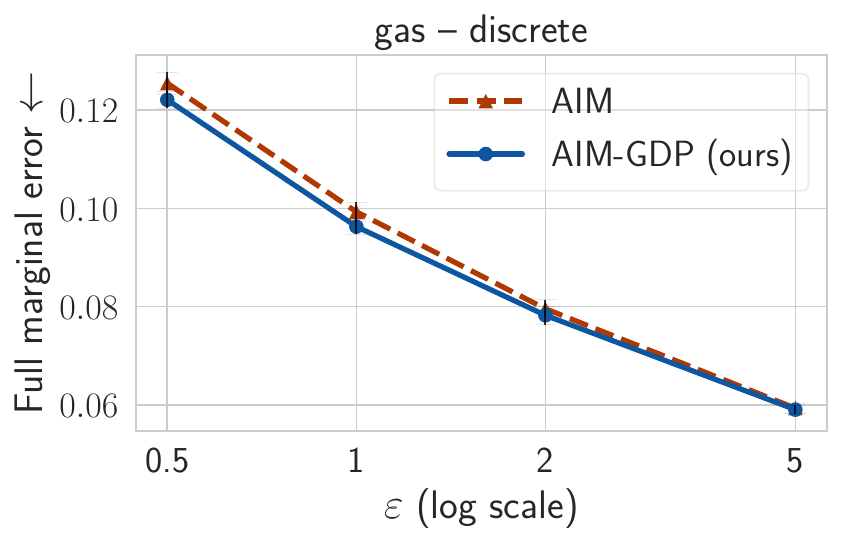}
    \caption{\aimgdp vs. \aim across full marginal error across four datasets in a discrete pipeline. We report 95\% bootstrap CI over the replications.}
    \label{fig:aimgdp_full}
\end{figure}

\begin{figure}[t!]
    \centering
    \includegraphics[width=0.45\linewidth]{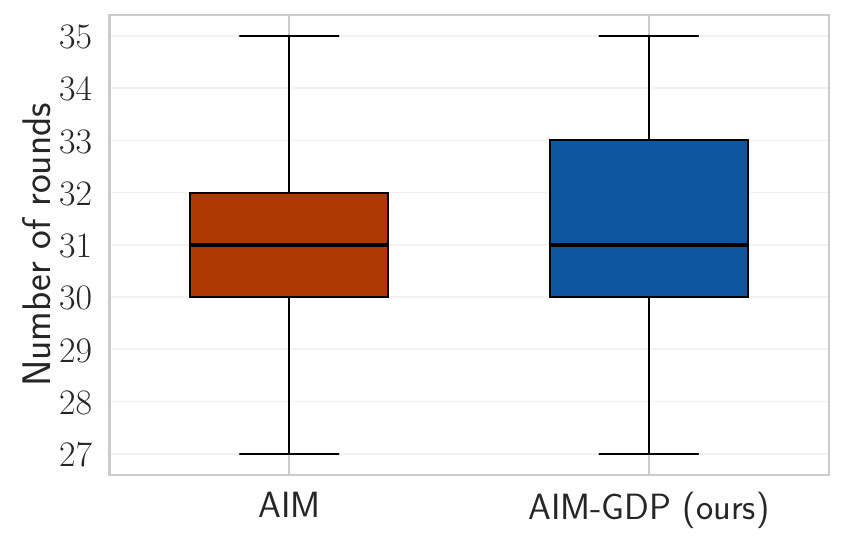}
    \hfill
    \includegraphics[width=0.45\linewidth]{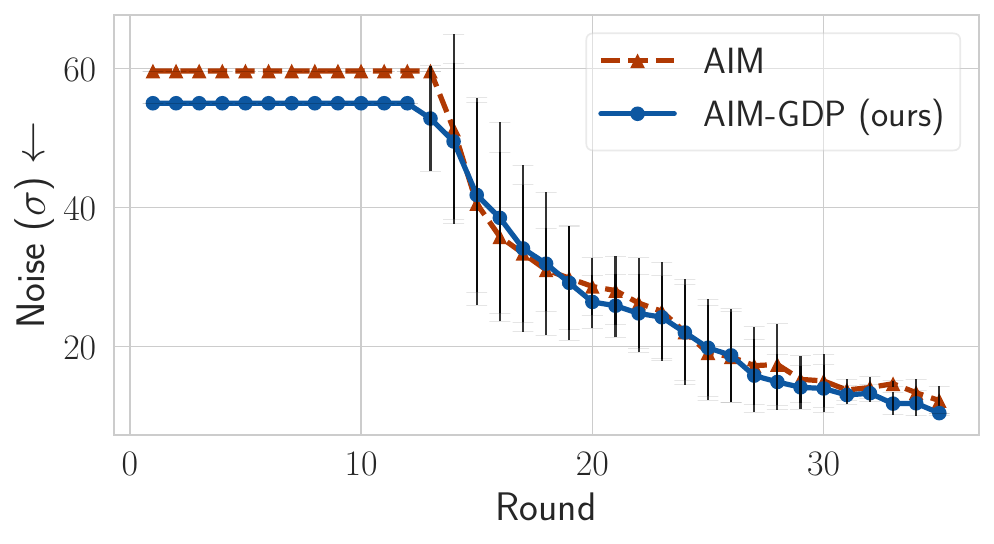}
    \caption{\aimgdp vs. \aim across number of rounds and noise ($\sigma$) per round averaged across 25 runs on Adult in a discrete pipeline.}
    \label{fig:aimgdp_diagnostics}
\end{figure}

\begin{figure}[t!]
    \centering
    \begin{subfigure}{0.90\linewidth}
        \centering
        \includegraphics[width=\linewidth]{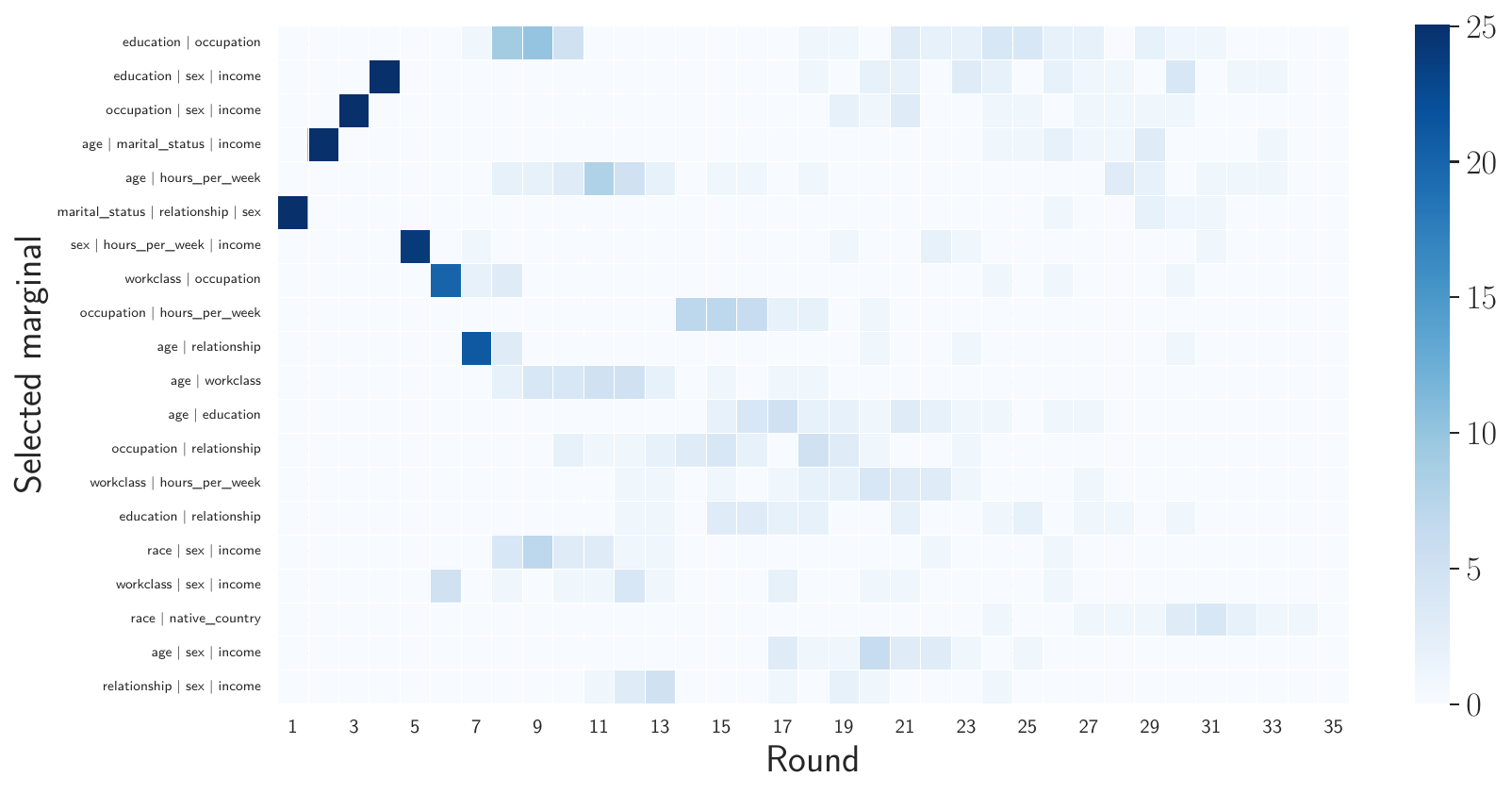}
        \caption{\aim}
        \label{fig:aim-marginals}
    \end{subfigure}
    \vspace{1em}
        \begin{subfigure}{0.90\linewidth}
        \centering
        \includegraphics[width=\linewidth]{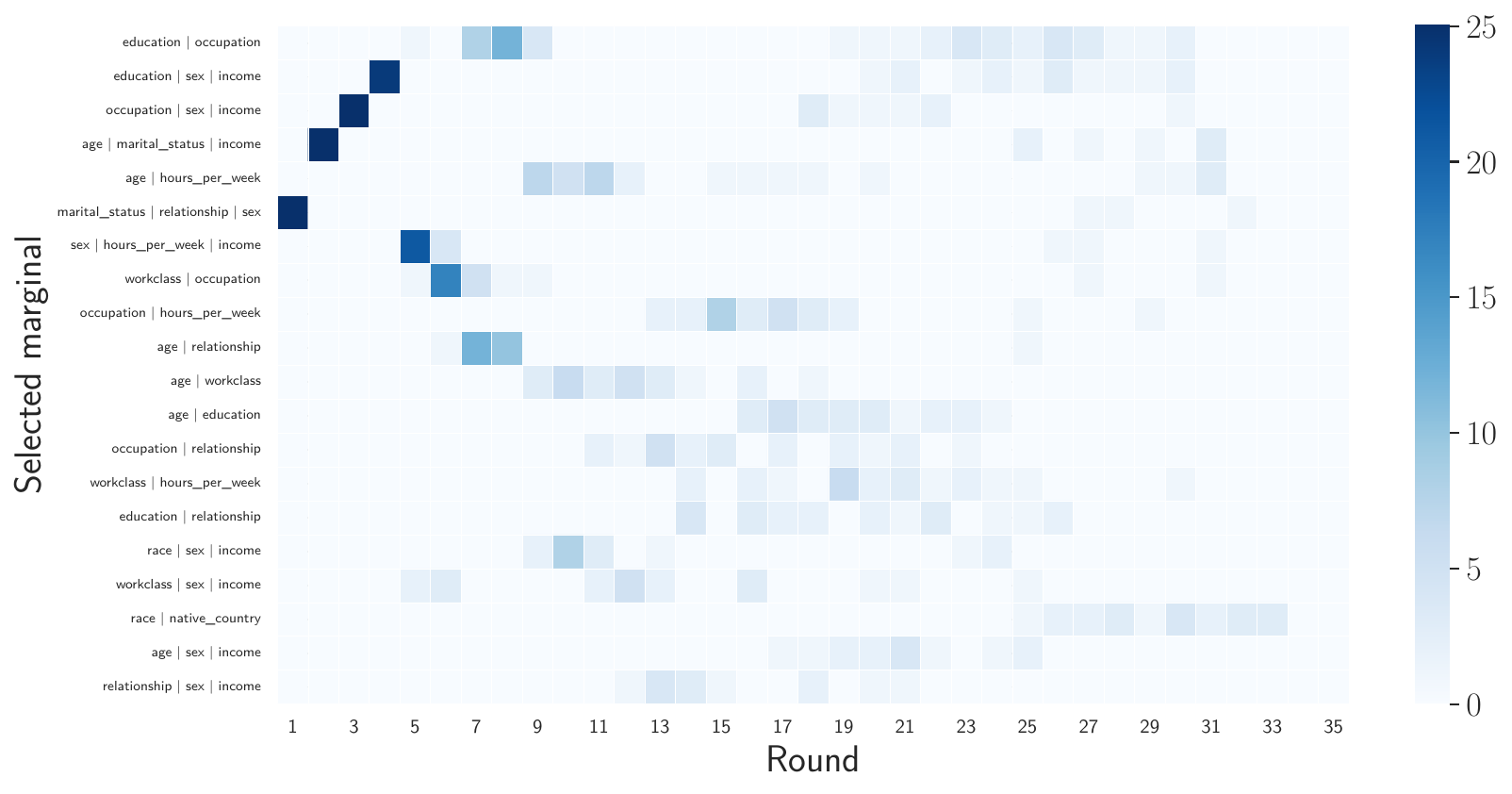}
        \caption{\aimgdp (ours)}
        \label{fig:aim-gdp-marginals-diagnostics}
    \end{subfigure}
    \vspace{1em}
    \begin{subfigure}{0.66\linewidth}
        \centering
        \includegraphics[width=\linewidth]{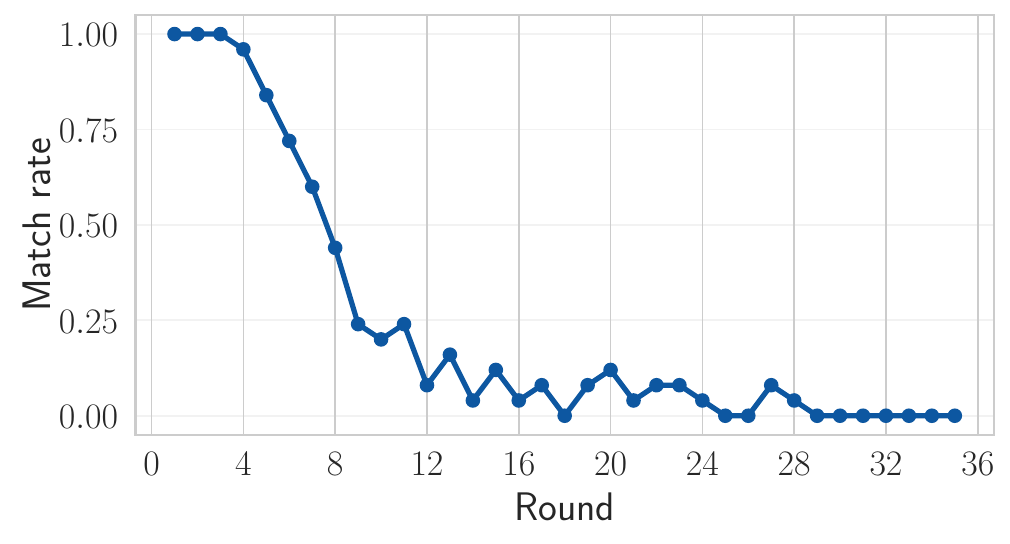}
        \caption{Selected marginal match rate}
        \label{fig:aim-gdp-marginals-match}
    \end{subfigure}
    \caption{\aimgdp vs. \aim across the selected marginal per round averaged across 25 runs on Adult in a discrete pipeline.}
    \label{fig:aimgdp_network}
\end{figure}

\subsubsection{Attack-Aware Calibration}
In the standard approach to calibrating DP mechanisms, we start with a target $(\varepsilon^\star, \delta^\star)$-DP guarantee for a small $\delta^\star$, usually smaller than the size of the dataset~\cite[see, e.g.,][]{ponomareva2023dp}. We can then evaluate the risk due to MIAs~\cite{kulynych2024attack} or other attacks~\cite{kulynych2025unifying} by converting $(\varepsilon^\star, \delta^\star)$-DP guarantee to $f$-DP via a standard conversion~\cite{dong2019gaussian}:
\begin{equation}\label{eq:dp-to-f}
    f_{\varepsilon, \delta}(\alpha) = \max\{0, 1 - e^{\varepsilon}\alpha - \delta, e^{-\varepsilon}(1 - \alpha - \delta)\}.
\end{equation}

Respectively, when we want to calibrate $\varepsilon^\star$ to a given level of risk in terms of $f$-DP, i.e., $(\alpha^\star, \beta^\star)$, we solve \cref{eq:dp-to-f} for a given small $\delta^\star$. We refer to this as \emph{standard risk calibration}.

In \emph{attack-aware calibration}~\cite{kulynych2024attack}, we start with the $f$-DP target $(\alpha^\star, \beta^\star)$, and calibrate the parameters of the DP mechanism directly to satisfy the $f$-DP level. 
In the case of a GDP mechanism, we can then solve the direct $f$-DP expression to derive the $\mu$ parameter that satisfies the target $(\alpha^\star, \beta^\star)$. In the case of PLD analysis, we can can run search, e.g., bisection, for the mechanism parameters until the $f$-DP curve dervied from the PLD~\cite{kulynych2024attack} satisfies the target requirement.

\subsubsection{Auditing}
DP auditing empirically tests whether the privacy leakage of an implementation is consistent with its theoretical guarantee~\cite{nasr2023tight,annamalai2024what}.
Following \citet{ganev2026tight}, we formulate auditing as a hypothesis-testing problem: a challenger repeatedly runs the mechanism on two neighboring datasets, and an adversary attempts to distinguish which dataset was used from the mechanism's output.
The resulting FPR and TPR define an empirical point on the
privacy trade-off curve.
As \aimgdp directly satisfies GDP, we convert this point into a statistical lower bound $\mu_{\mathrm{emp}}$ on the privacy leakage and compare it directly with the theoretical $\mu$.

We use the worst-case neighboring datasets $D_{\mathrm{out}}=\{[0,0,0]\}^{10}$ and $D_{\mathrm{in}}=D_{\mathrm{out}}\cup\{[1,1,1]\}$, chosen to maximize the target record's influence on the measured marginals.
We train 10,000 independent instances of \aimgdp, equally split between $D_{\mathrm{out}}$ and $D_{\mathrm{in}}$, and generate synthetic records from each model.
Unlike prior audits of \aim, we run the full \aimgdp model, including its adaptive/data-dependent marginal selection.

We consider black-box attacks using features computed from the released synthetic data, white-box attacks using the mechanism's internal measurements, and an ensemble using both.
We split the 10,000 labeled runs into 4,000 training,
2,000 validation, and 4,000 test examples, train an XGBoost classifier~\cite{chen2016xgboost}, and select the decision threshold on the validation set to maximize attack advantage ($\mathrm{TPR}-\mathrm{FPR}$).
We evaluate the selected attack once on the held-out test set and use the resulting FPR/TPR with the Bayesian procedure~\cite{zanella2023bayesian, ganev2026tight} to obtain a 90\%-posterior lower bound $\mu_{\mathrm{emp}}$.
The resulting audit is discussed in \cref{sec:exp-aimgdp}.

\subsection{End-to-End Synthetic Data with \aimgdp, DP preprocessing, and \mixcal}
\label{app:emp_c}

We show detailed results for all four (original/mixed-type) datasets in an end-to-end pipeline in \cref{fig:aimgdp_full_e2e}.
Additionally, we evaluate the pipeline under Poisson subsampling in simulation (\cref{fig:subsampled_accounting}).
We discuss these in \cref{sec:exp-aimgdp-e2e}.

\begin{figure}[t!]
    \centering
    \includegraphics[width=0.45\linewidth]{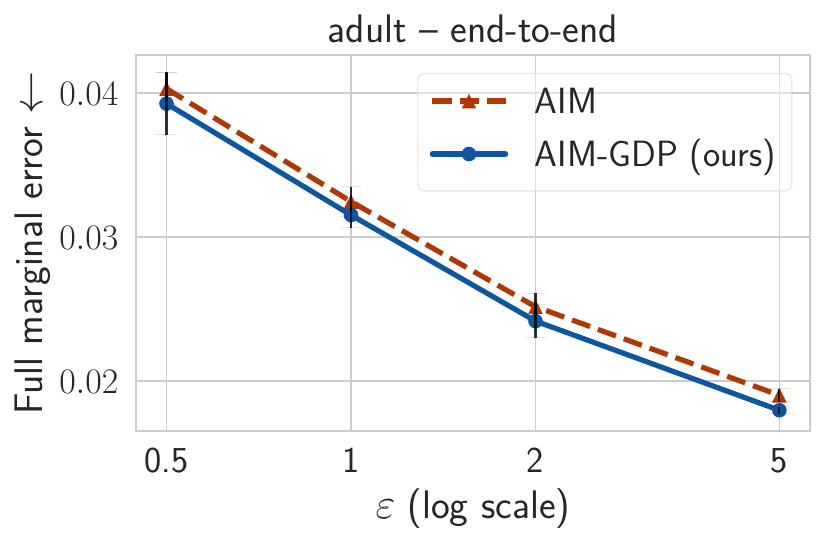}
    \vspace{0.5em}
    \includegraphics[width=0.45\linewidth]{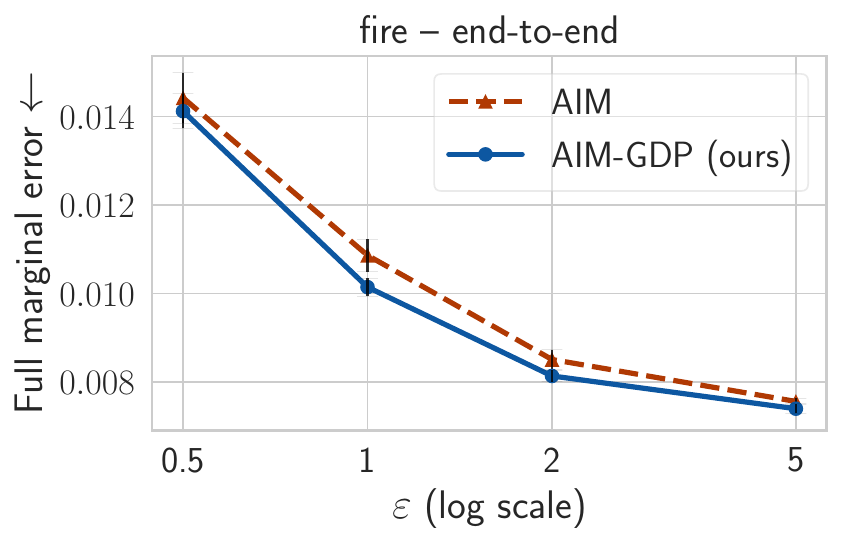}
    \includegraphics[width=0.45\linewidth]{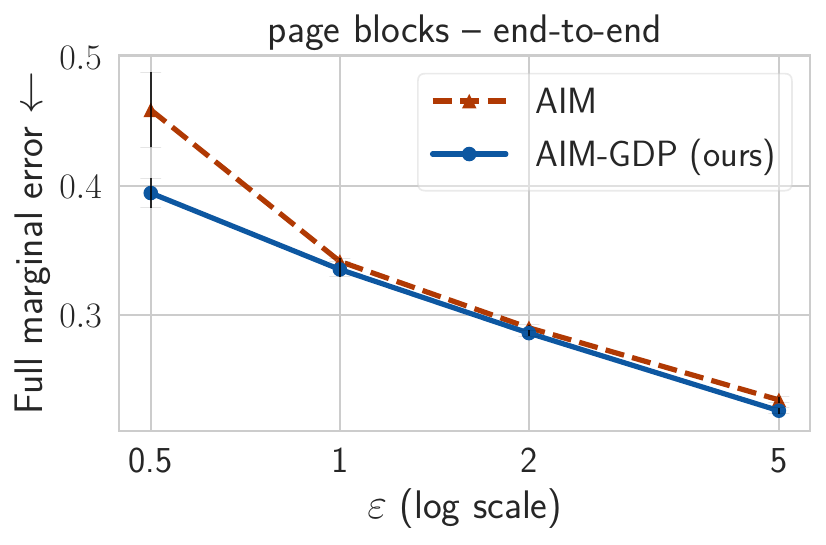}
    \vspace{0.5em}
    \includegraphics[width=0.45\linewidth]{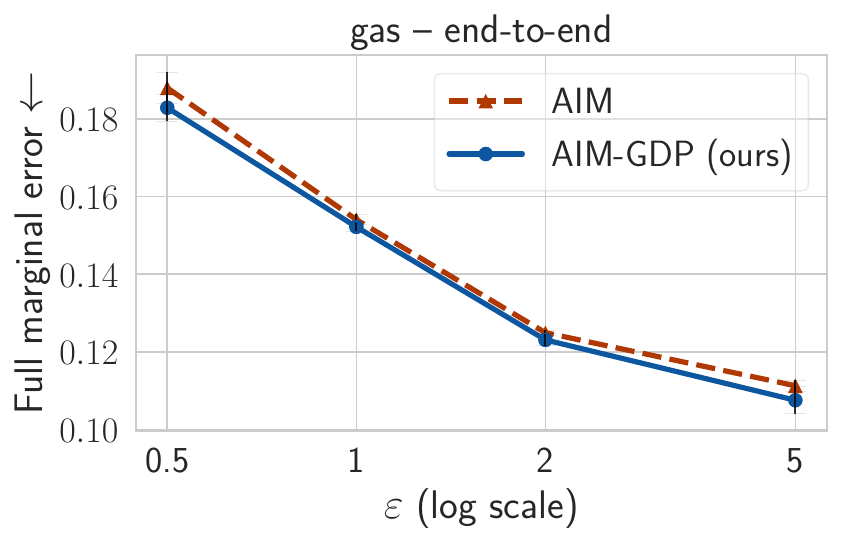}
    \caption{\aimgdp vs. \aim across full marginal error across four datasets in an end-to-end pipeline. We report 95\% bootstrap CI over the replications.}
    \label{fig:aimgdp_full_e2e}
\end{figure}

\begin{figure}[t!]
\centering
\includegraphics[width=0.55\textwidth]{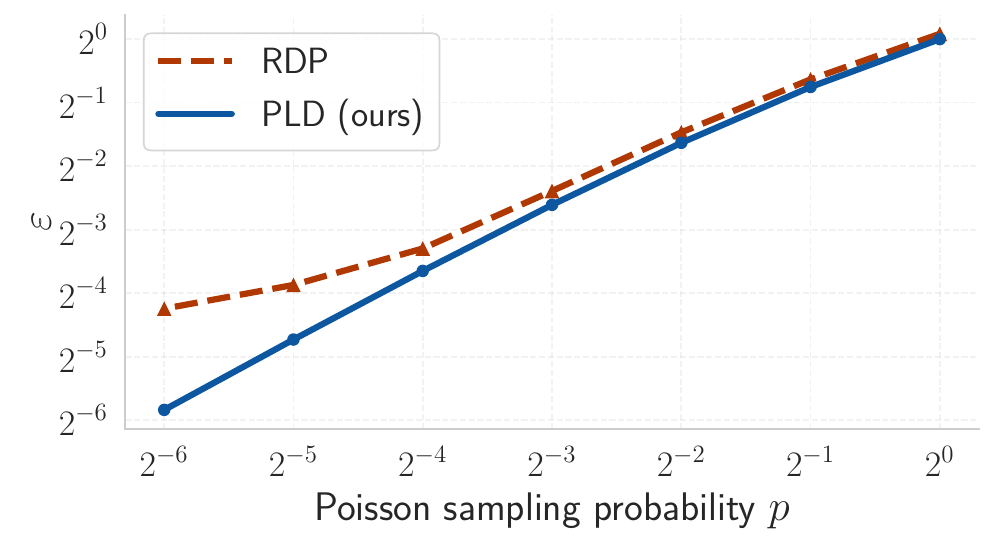}
\caption{DP $\varepsilon$ at $\delta = 10^{-6}$ across Poisson sampling probabilities $p = 2^{-k}$ ($k = 0, \dots, 6$), with $\rho$ calibrated to an un-amplified $\varepsilon = 1.0$ at $p = 1$.}
\label{fig:subsampled_accounting}
\end{figure}

\subsection{Generating Synthetic Text Data with LLMs and Tabular Conditioning}\label{app:actg}

We follow the conditional generation approach of \citet{hu2026actg}, which we summarize here. The goal is to generate a synthetic version of the \emph{bioRxiv} dataset. First, \citeauthor{hu2026actg} fed each abstract in the \emph{bioRxiv} dataset to Gemini-2.5-flash-lite to extract structured features. In particular, the following 8 features were extracted: primary-research-area, model-organism, 
experimental-approach, dominant-data-type, research-focus-scale, disease-mention, sample-size, research-goal. Each feature is categorical. See Figure 7 in \citet{hu2026actg} for the allowed values of each category. The result is 8 categorical features for each abstract in the dataset. 

These features are treated as sensitive, hence \aim is run on them to generate a list of synthetic features. The next step is to condition a DP-finetuned language model on these privatized features. Here, Gemma 3 1B is fine tuned with differential privacy on the feature/abstract pairs. We used a batch size of 448 and iterations $T = 512$ (8 epochs). We set clipping norm $c$ to 1. We tune the learning rate $
\eta \in \{\num{1e-4}, \num{3e-4}, \num{1e-3}\}$, with $\num{3e-4}$ converging the fastest. We used a cosine decay learning rate schedule. We perform LoRA finetuning with a LoRA rank of 8, $\alpha = 16$ and dropout of 0.05. Training took 3 hours on 2 A100 GPUs, where data parallel was used to parallelize training. Since the 95\% quantile of the abstracts in bioRxiv had 512 tokens, we used a context length of 512. We use the MAUVE scores~\cite{pillutla2023mauve} as the main evaluation metric for this demonstration.

We remark that we trained for a tenth as long as \citeauthor{hu2026actg} did, as they trained for 80 epochs whereas we trained for 8. Moreover, we used a smaller learning rate of $\num{3e-4}$ whereas they used $\num{1e-3}$. While our evaluation loss values are within a tenth of \cite{hu2026actg}, our MAUVE scores are smaller than \cite{hu2026actg} reported scores, likely due to this reason. Namely, we report a MAUVE score of 60\% at $\varepsilon=1$ for the RDP baseline, which corresponds to vanilla ACTG, where as \cite{hu2026actg} reports a little under 80\%. Nevertheless, we are confident that, had we run for the full 80 epochs, we would still see improvement in our PLD approach since our approach allows for training at a smaller DP-SGD noise multiplier as reported in \cref{tab:actg_noise_multipliers}.

Once the model was fine tuned, we used nucleus decoding to sample conditioned on the features from \aim. We used temperature $T = 1.0, \text{top-p} = 0.95, \text{top-k} = 0$, and a batch size of 64. We sampled $N_\text{syn} = \num{5000}$ abstracts. Sampling took 3.5 hours on 2 A100 GPUs, where the parallelization was done by feeding half the features to condition on (2,500) to each GPU. 

The MAUVE score from \cref{fig:utility} was computed between the $\num{5000}$ synthetic examples sampled from the fine tuned LLM and the $\num{5000}$ real heldout examples from the test data from bioRxiv. MAUVE scores use an embedding model. Following \cite{hu2026actg}, we used Spector2~\cite{singh2023scirepeval}, a specialized sentence embedding model finetuned on abstracts of scientific papers. \cite{hu2026actg} showed that this domain-specific embedding model improved the reliability of the MAUVE scores. 

We repeated the ACTG training pipeline twice, once using  \citeauthor{hu2026actg}'s original accounting approach, which used RDP accounting with \aim, and our proposed technique, which uses \aimgdp and PLD accounting. The \mixcal algorithm was used in the following manner. We discuss the PLD approach here. The RDP approach follows analogously. Suppose our target is for the composition of \aimgdp and DP-SGD to satisfy $(\varepsilon^\star, \delta^\star)$-DP. Given a starting dummy budget $(\varepsilon_\text{dummy}, \delta_\text{dummy})$, we allocate 71\% to \aim and 29\%  to DP-SGD following \cite{hu2026actg}. 
By this, we mean we calibrate \aimgdp to satisfy $(0.71 \cdot \varepsilon_\text{dummy}, 0.71 \cdot \delta_\text{dummy})$-DP, and DP-SGD to satisfy $(0.29 \cdot \varepsilon_\text{dummy}, 0.29 \cdot \delta_\text{dummy})$-DP. We then run PLD composition on the mechanism. As mentioned in the main body,  the resulting PLD $\varepsilon$ will larger than our target $\varepsilon^\star$, hence we run a binary search over the dummy value of $\varepsilon_\text{dummy}$ until the composed mechanisms under PLD  composition satisfies our target $(\varepsilon^\star, \delta^\star)$.

\end{appendices}

\end{document}